\documentclass[12pt]{article}

\usepackage[letterpaper,top=1in,bottom=1in,left=1in,right=1in]{geometry}

\usepackage{upgreek}

\usepackage{titlesec}
\titleformat{\section}{\normalfont\scshape\large}{\S\thesection.}{0.5em}{}
\titleformat{\subsection}{\normalfont\scshape\normalsize}{\thesubsection.}{0.5em}{}
\titleformat{\subsubsection}{\normalfont\itshape\normalsize}{\thesubsubsection.}{0.5em}{}

\usepackage{amsthm}
\newtheorem{theorem}{Theorem}[section]
\newtheorem{lemma}[theorem]{Lemma}

\newtheorem{corollary}[theorem]{Corollary}

\theoremstyle{definition}
\newtheorem{definition}[theorem]{Definition}
\newtheorem{example}[theorem]{Example}

\theoremstyle{remark}
\newtheorem{remark}[theorem]{Remark}

\usepackage[numbers]{natbib}
\title{Nested Sequents for Intuitionistic Multi-Modal Logics: Modularity, Cut-Elimination, and Undecidability}
\author{Tim S. Lyon\\
\small Technische Universit{\"a}t Dresden\\
\small \texttt{timothy\_stephen.lyon@tu-dresden.de}}
\date{}

\usepackage{mathrsfs}
\usepackage{synttree, bussproofs,scrextend, rotating, xcolor, upgreek}
\usepackage[ruled, linesnumbered]{algorithm2e}

\usepackage{esvect,pgf, tikz, color}
\usetikzlibrary{arrows, automata}
\usetikzlibrary{trees}
\usetikzlibrary{fit}
\usepackage{tikz}
\usetikzlibrary{positioning,arrows,calc,decorations.markings}
\tikzset{
modal/.style={>=stealth',shorten >=1pt,shorten <=1pt,auto,node distance=1.5cm,semithick},
world/.style={circle,draw,minimum size=0.5cm,fill=gray!15},
point/.style={circle,draw,inner sep=0.5mm,fill=black},
reflexive above/.style={->,loop,looseness=7,in=120,out=60},
reflexive below/.style={->,loop,looseness=7,in=240,out=300},
reflexive left/.style={->,loop,looseness=7,in=150,out=210},
reflexive right/.style={->,loop,looseness=7,in=30,out=330}
}

\usepackage[all]{xy}
\usepackage{changepage,enumerate,proof,upgreek,relsize,comment}
\usepackage{multicol}
\usepackage{trfsigns}

\usepackage{wrapfig}
\usepackage{extpfeil}

\DeclareSymbolFont{extraup}{U}{zavm}{m}{n}
\DeclareMathSymbol{\vardiamond}{\mathalpha}{extraup}{87}

\makeatletter
\newcommand*{\owedge}{%
  \mathbin{%
    \mathpalette\@owedge{}%
  }%
}
\newcommand*{\@owedge}[2]{%
  \sbox0{$#1\oplus\m@th$}%
  \dimen2=.5\dimexpr\wd0-\ht0-\dp0\relax 
  \dimen@=\dimexpr\ht0+\dp0\relax
  \def\lw{.04}
  \def\radius{.5-\lw/2}%
  \kern\dimen2 
  \tikz[
    line width=\lw\dimen@,
    line join=round,
    x=\dimen@,
    y=\dimen@,
    baseline=\dimexpr-.5\dimen@+\dp0\relax,
  ]
  \draw
    (0,0) circle[radius=\radius]
    (225:\radius) -- (0,.5-\lw) -- (-45:\radius)
  ;%
  \kern\dimen2 
}
\makeatother

\definecolor{tim}{RGB}{0, 0, 250}

\newcommand{\ifandonlyif}{\textit{iff} }
\newcommand{\iffi}{\textit{iff} }
\newcommand{\etc}{$\ldots$ }
\newcommand{\dfn}{Definition}
\newcommand{\fig}{Figure}
\newcommand{\lem}{Lemma}
\newcommand{\thm}{Theorem}

\newcommand{\sect}{Section}
\newcommand{\cptr}{Ch.}

\newcommand{\D}{\mathrm{D}_{\charx}}
\newcommand{\T}{\mathrm{T}_{\charx}}
\newcommand{\four}{\mathrm{4}_{\charx}}
\newcommand{\B}{\mathrm{B}_{\charx}}
\newcommand{\five}{\mathrm{5}_{\charx}}
\newcommand{\ipa}{\mathrm{IPA}}
\newcommand{\axs}{\mathcal{A}}

\newcommand{\bl}{[}
\newcommand{\br}{]}
\newcommand{\sbl}{\{}
\newcommand{\sbr}{\}}

\newcommand{\cglcalc}{\mathsf{NK_{m}}(\axs)}

\newcommand{\calc}{\mathsf{NIK_{m}}(\axs)}

\newcommand{\ns}{\mathcal{G}}
\newcommand{\nsii}{\mathcal{H}}
\newcommand{\nsiii}{\mathcal{K}}
\newcommand{\nsiv}{\mathcal{J}}

\newcommand{\empseq}{\emptyset}

\newcommand{\sar}{\vdash}

\newcommand{\ent}{\vdash_{\!\axs}}
\newcommand{\sat}{\vDash} 
\newcommand{\asat}{\vDash_{\!\axs}}

\newcommand{\da}{\downarrow}

\newcommand{\h}{\mathsf{H}}

\newcommand{\prf}{\pi} 
            
\newcommand{\km}{\mathsf{K_{m}}}

\newcommand{\ikm}{\mathsf{IK_{m}}}
\newcommand{\ik}{\mathsf{IK}}
\newcommand{\ikt}{\mathsf{IKt}}
\newcommand{\ikma}{\mathsf{IK_{m}}(\axs)}

\newcommand{\cate}{}
\newcommand{\concat}{\cdot}

\newcommand{\g}[1]{g(#1)}

\newcommand{\glang}{L_{\g{\axs}}}

\newcommand{\pto}{\longrightarrow}

\newcommand{\empstr}{\varepsilon} 

\newcommand{\charx}{x}
\newcommand{\chary}{y}
\newcommand{\charz}{z}
\newcommand{\chara}{a}
\newcommand{\charb}{b}
\newcommand{\charc}{c}
\newcommand{\stra}{s}
\newcommand{\strb}{t}
\newcommand{\strc}{r}
\newcommand{\conv}[1]{\overline{#1}}
\newcommand{\albetstr}{\albet^{\ast}}

\newcommand{\len}[1]{|#1|}

\newcommand{\strabox}{[\stra]}
\newcommand{\stradia}{\langle \stra \rangle}

\newcommand{\osdr}{\pto_{\g{\axs}}}
\newcommand{\dr}{\pto_{\g{\axs}}^{*}}

\newcommand{\iimp}{\rightarrow} 
\newcommand{\imp}{\rightarrow}
\newcommand{\ieq}{\leftrightarrow} 
\newcommand{\inot}{\neg}

\newcommand{\xbox}{[\charx]}

\newcommand{\xdia}{\langle \charx \rangle}

\newcommand{\xboxc}{[\conv{\charx}]}
\newcommand{\xdiac}{\langle \conv{\charx} \rangle}

\newcommand{\dia}{\diamondsuit}
\newcommand{\diap}[1]{\langle #1 \rangle}
\newcommand{\boxp}[1]{[ #1 ]}

\newcommand{\albet}{\Upsigma}
\newcommand{\albetf}{\albet^{+}}
\newcommand{\albetb}{\albet^{-}}
\newcommand{\lang}{\mathscr{L}}

\newcommand{\prop}{\mathsf{Prop}}

\newcommand{\id}{\mathsf{id}}

\newcommand{\xboxr}{\xbox\mathsf{R}}
\newcommand{\xboxl}{\xbox\mathsf{L}}
\newcommand{\xdiar}{\xdia\mathsf{R}}
\newcommand{\xdial}{\xdia\mathsf{L}}

\newcommand{\conl}{\land\mathsf{L}}
\newcommand{\conr}{\land\mathsf{R}}
\newcommand{\disl}{\lor\mathsf{L}}
\newcommand{\disr}{\lor\mathsf{R}}

\newcommand{\iimpl}{{\iimp}\mathsf{L}}

\newcommand{\iimpr}{{\iimp}\mathsf{R}}
\newcommand{\dial}{\xdia\mathsf{L}}

\newcommand{\boxl}{\xbox\mathsf{L}}
\newcommand{\boxr}{\xbox\mathsf{R}}
\newcommand{\botl}{\bot\mathsf{L}}
\newcommand{\botr}{\bot\mathsf{R}}

\newcommand{\ddr}{\mathsf{d}_{\charx}}

\newcommand{\wkl}{\mathsf{wL}}
\newcommand{\wkr}{\mathsf{wR}}
\newcommand{\ew}{\mathsf{ew}}
\newcommand{\nec}{\mathsf{nec}_{\charx}}

\newcommand{\ec}{\mathsf{ec}}
\newcommand{\ctrl}{\mathsf{cL}}

\newcommand{\cut}{\mathsf{cut}}

\newcommand{\orth}[1]{{\Downarrow}{#1}}
\newcommand{\north}[1]{{\Uparrow}{#1}}

\newcommand{\names}{\mathtt{Names}}

\newcommand{\rcolii}{\mathcal{H}^{+}}

\newcommand{\nshole}{\{\cdot\}}

\newcommand{\tr}{\mathsf{tr}}

\newcommand{\mint}{\iota}
\newcommand{\frames}{\mathscr{F}}

\newcommand{\simplepath}[1]{\overset{\mathclap{#1}}{\leadsto}}
\newcommand{\prpath}[1]{\hspace{3pt}\raise-3pt\hbox{$\simplepath{#1}$}\hspace{3pt}}

\newcommand{\prgr}[1]{\mathsf{pg}(#1)}
\newcommand{\prv}{\mathsf{V}}
\newcommand{\pre}{\mathsf{E}}

\newcommand{\shift}{\mathsf{sft}(\axs)}
\newcommand{\set}[1]{\{#1\}}

\newcommand{\ru}{\mathsf{r}}

\newcommand{\dotprf}[1]{\;\;\mathbin{\rotatebox[origin=c]{-90}{$\overset{\mathbin{\rotatebox[origin=c]{90}{$#1$}}}{\hdots}$}}}

\newcommand{\proves}{\Vdash}

\newcommand{\branch}{\mathcal{B}}

\newcommand{\infs}[1]{\mathsf{InF}(#1)}
\newcommand{\outfs}[1]{\mathsf{OutF}(#1)}

\newcommand{\ddn}[1]{#1^{-}}

\begin{document}

\maketitle
\thispagestyle{empty}

\noindent
\begin{minipage}{\textwidth}
\noindent
\textbf{Abstract.} We introduce and study single-conclusioned nested sequent calculi for a broad class of intuitionistic multi-modal logics known as \emph{intuitionistic grammar logics (IGLs)}. These logics serve as the intuitionistic counterparts of classical grammar logics, and subsume standard intuitionistic modal and tense logics, including $\mathsf{IK}$ and $\mathsf{IKt}$ extended with combinations of the $\mathrm{T}$, $\mathrm{B}$, $\mathrm{4}$, $\mathrm{5}$, and $\mathrm{D}$ axioms. We analyze fundamental invertibility and admissibility properties of our calculi and introduce a novel structural rule, called the \emph{shift} rule, which unifies standard structural rules arising from modal frame conditions into a single rule. This rule enables a purely syntactic proof of cut-admissibility that is uniform over all IGLs, and yields completeness of our nested calculi as a corollary. Finally, we define a negative translation that constitutes a faithful embedding of classical grammar logics (CGLs) into IGLs, witnessed by proof transformations between multi-conclusioned and single-conclusioned nested sequent proofs for CGLs and IGLs, respectively. This reduces the general validity problem for CGLs to that of IGLs. The general validity problem over a class $\mathcal{C}$ of logics asks: given a logic $\mathsf{L} \in \mathcal{C}$ and a formula $A$, is $A$ valid in $\mathsf{L}$? As this problem is known to be undecidable for CGLs, our reduction implies its undecidability for IGLs as well.
\end{minipage}

\section{Introduction}\label{sec:introduction}










\subsection{Intuitionistic Modal Logics.} Numerous authors have proposed logics that provide an intuitionistic account of modal reasoning~\cite{BovDov84,Bul65,Dov85,Fit48,PloSti86,Fis77,Fis84,Sim94}; see Stewart et al.~\cite{StePaiAle15} for a survey. Among these, the formulations introduced by Fischer-Servi~\cite{Fis77,Fis84} and by Plotkin and Stirling~\cite{PloSti86} have gained particular prominence, especially through the work of Simpson~\cite{Sim94}, who placed such systems on a firm philosophical foundation. The central idea, due to Simpson, is to interpret modalities from within an intuitionistic meta-theory. This perspective transfers core intuitionistic principles to the modal setting, such as the disjunction property, the failure of the law of excluded middle, and breaks the classical duality between 
necessity and possibility. Beyond their theoretical interest, intuitionistic modal logics (IMLs) have proven useful in computer science, being used to design verification techniques~\cite{FaiMen95}, in reasoning about functional programs~\cite{Pit91}, and in the definition of programming languages~\cite{DavPfe01}.

In~\cite{Lyo21b}, multi-modal generalizations of these logics were introduced under the name \emph{intuitionistic grammar logics (IGLs)}. This family of logics subsumes the logics by Fischer-Servi~\cite{Fis84}, Plotkin and Stirling~\cite{PloSti86}, and Ewald’s intuitionistic tense logics (ITLs) which include temporal modalities referring to past and future states of affairs~\cite{Ewa86}. IGLs are also closely related to the intuitionistic description logic $\mathsf{iALC}$~\cite{HaePaiRad11,HaeRad14}, used for modeling and querying legal ontologies. While $\mathsf{iALC}$ can be viewed as a multi-modal variant of $\mathsf{IK}$ (cf.~\cite{Fis84,PloSti86}), the class of IGLs constitutes a more expansive generalization supporting converse modalities, seriality, and intuitionistic path axioms (IPAs).

The standard modalities in IMLs (cf.~\cite{Fis84,PloSti86,Sim94}) refer only to future states in a model. One way in which IGLs generalize IMLs is by also including converse modalities, which refer to past states in models. This gives IGLs a natural temporal interpretation and is analogous to the inclusion of inverse roles in description logics (cf.~\cite{BaaHorLutSat17,HorSat04}). Path axioms express that a sequence of modal transitions implies a direct connection between the initial and terminal worlds; for example, the frame condition $w_{0} R_{\charx_{1}} w_{1}, \ldots, w_{n-1} R_{\charx_{n}} w_{n} \Rightarrow w_{0} R_{\charx} w_{n}$ is expressible by an IPA. Such axioms capture standard frame conditions in modal logic (e.g., reflexivity, symmetry, transitivity, Euclideanity) and serve as the analog of role inclusion axioms in description logics. Consequently, IGLs can be viewed as intuitionistic counterparts of both the description logic $\mathcal{RI}$~\cite{HorSat04} and \emph{classical grammar logics (CGLs)}~\cite{DemNiv05,TiuIanGor12}. Therefore, IGLs provide a rich and unifying language that subsumes many intuitionistic (multi-)modal systems within a single framework.

\subsection{Proof Theory and Sequents.} It is standard practice to equip a logic with an appropriate proof system that supports the study of its (meta-)logical properties and enables formal reasoning. Typically, one seeks to supply \emph{analytic} proof systems for logics, which are advantageous from a computational standpoint. This is because such systems satisfy the \emph{subformula property}--every formula occurring in a derivation is a subformula of the conclusion. 
Since Gentzen’s seminal work~\cite{Gen35a,Gen35b}, the sequent formalism has become a popular framework for the construction of analytic, or cut-free, proof systems. Such systems are particularly useful for establishing non-trivial results such as interpolation~\cite{Mae60,FitKuz15} and decidability~\cite{Bru09,Sim94}, and they facilitate automated reasoning~\cite{Dyc92,HaePaiRad11} by supporting systematic proof-search and the extraction of derivations from formulae.

At present, IGLs are only equipped with Hilbert systems~\cite{Lyo20a}, which are sub-optimal as they invalidate the sub-formula property. We address this gap by introducing analytic sequent-style calculi for IGLs, formulated in the \emph{nested sequent} framework. This choice is motivated by the fact that traditional Gentzen sequents are not expressive enough to yield natural, cut-free systems for logics with converse modalities.

The nested sequent formalism is a generalization of Gentzen's sequent formalism and a nested sequent is a tree of Gentzen sequents. The creation of this formalism is often attributed to Bull~\cite{Bul92} and Kashima~\cite{Kas94}, though later work by Br\"unnler~\cite{Bru09} and Poggiolesi~\cite{Pog09} was instrumental to the popularity of the formalism.\footnote{Leivant~\cite[p.~361]{Lei81} had already introduced a notational variant of nested sequents in his proof-theoretic work on propositional dynamic logic, where formulae are prefixed by so-called execution sequences.} Nested sequents yield elegant calculi that minimize syntactic overhead, produce compact proofs, and where termination of proof-search is more easily established (cf.~\cite{Lyo21thesis,LyoOst23}). Moreover, such systems have proven well-suited in extracting counter-models from failed proof-search~\cite{Bru09,TiuIanGor12}, and they usually admit syntactic cut-admissibility. Thanks to these aesthetic and computational advantages, nested sequent systems have been applied in a variety of domains, including knowledge integration~\cite{LyoGom22}, computing explicit definitions in description logics~\cite{LyoKar24}, 
and in the study of axiomatizability~\cite{IshKik07}. See Lellmann and Poggiolesi~\cite{LelPog24} for a comprehensive survey.\\

\subsection{Related Work and Contributions.} The proof theory for standard IMLs--that is, those based on the work of Fischer-Servi~\cite{Fis84} and Plotkin and Stirling~\cite{PloSti86}--is diverse, encompassing Hilbert-style systems~\cite{Fis84,PloSti86}, labeled natural deduction and sequent calculi~\cite{HaePaiRad11,MarMorStr21,Sim94}, as well as nested sequent systems~\cite{KuzStr19,Lyo21a,Str13}. Since the present work is primarily concerned with the sequent formalism, we focus here on developments in sequent-style calculi.

Among these, Simpson~\cite{Sim94} introduced labeled sequent calculi for (graph-consistent) extensions of $\mathsf{IK}$ with geometric frame conditions,\footnote{See Simpson~\cite[p.~155]{Sim94} for a discussion of graph consistency.} while Marin et al.~\cite{MarMorStr21} proposed fully labeled calculi for extensions of $\mathsf{IK}$ with intuitionistic Scott–Lemmon axioms. Although both sets of proof systems are formulated within the labeled sequent paradigm, they differ in structure. Simpson’s systems capture intuitionistic reasoning \emph{à la} Gentzen by restricting the consequent of sequents to a single formula. By contrast, Marin et al.’s systems omit the single-conclusion restriction and instead internalize the intuitionistic accessibility relation directly into the syntax of labeled sequents. This increases the syntactic complexity of Marin et al.’s systems, but has the advantageous effect that all rules become invertible--this is precluded when the single-conclusion restriction is enforced.  We note that all such systems are cut-free, but only Marin et al. provide a syntactic cut-admissibility algorithm for their systems.


In contrast to the labeled approach, several works have developed cut-free nested sequent systems for IMLs. Stra{\ss}burger~\cite{Str13} introduced nested systems for the ``intuitionistic modal cube,'' namely, extensions of $\mathsf{IK}$ with combinations of the axioms for reflexivity, symmetry, transitivity, Euclideanity, and seriality, denoted $\mathrm{T}$, $\mathrm{B}$, $\mathrm{4}$, $\mathrm{5}$, and $\mathrm{D}$, respectively.\footnote{Stra{\ss}burger’s~\cite{Str13} systems can be viewed as intuitionistic counterparts of Brünnler’s nested systems for classical modal logics~\cite{Bru09}.} These were later generalized in~\cite{Lyo21a} to handle $\mathsf{IK}$ extended with Horn–Scott–Lemmon axioms and seriality. Like Simpson’s labeled systems, these nested calculi are single-conclusioned. Kuznets and Stra{\ss}burger~\cite{KuzStr19} further relaxed this restriction by allowing multi-conclusioned nested sequents in proofs, and only restricting to single-conclusioned sequents in applications of right implication and box rules.


We design our nested sequent systems for IGLs in the single-conclusioned style of~\cite{Lyo21a,Str13}.\footnote{While~\cite{Lyo21a,Str13} use input and output labels to distinguish antecedent and consequent formulae, we employ a more standard two-sided notation in this paper.} This is because the class of IGLs introduced in~\cite{Lyo21b} can be naturally captured within the nested sequent setting, which is preferable to the labeled approach. Nested systems are generally more economical, admitting shorter proofs with less syntactic overhead than labeled sequent calculi (cf.~\cite{Lyo21thesis,LyoOst23}).

Our contributions in this paper are discussed below:

(1) We initiate the structural proof theory of IGLs by introducing their first nested sequent calculi. We establish soundness and completeness and show that these calculi are well-behaved, enjoying fundamental properties such as the admissibility of structural rules and invertibility of most left rules. A particular novelty of our systems is their \emph{modularity}, that is, it is simple to transform a nested system for one IGL into a system for another by systematically changing the side conditions on certain rules. This modularity enables the formulation of a new structural rule, the \emph{shift} rule $\shift$, which uniformly captures all frame conditions induced by IPAs and which we prove height-preserving admissible.

This feature distinguishes our framework from existing nested sequent systems for (intuitionistic) modal logics, which typically require additional structural rules to capture frame conditions in a complete way (cf.~\cite{Bru09,Str13}). Exploiting the admissibility of $\shift$, we present a uniform syntactic proof of cut-admissibility that covers all IGLs. This is yet another contribution of this work, as existing cut-admissibility proofs for modal systems tend to rely on ad hoc structural and cut rules to eliminate cuts in special cases~\cite{Bru09,LyoOrl23,Str13}. Our approach avoids these ad hoc considerations altogether and thereby streamlines the cut-elimination argument.

(2) Our systems serve as the natural intuitionistic counterparts of the nested calculi for classical grammar logics (CGLs) introduced by Tiu et~al.~\cite{TiuIanGor12}, in much the same way that Stra{\ss}burger’s intuitionistic systems~\cite{Str13} correspond to Brünnler’s classical systems~\cite{Bru09}. Tiu et~al. proved cut-admissibility for their nested calculi by means of an \emph{indirect} method, viz., a nested sequent proof is translated into a display calculus proof, cuts are eliminated there, and then the resulting display proof is translated back into a cut-free nested sequent proof. Since our nested systems are essentially single-conclusioned variants of Tiu et~al.'s systems, by making our nested systems multi-conclusioned, our cut-admissibility algorithm can be easily converted into a \emph{direct} algorithm that eliminates cuts for CGLs.

(3) The \emph{general validity problem} over a class $\mathcal{C}$ of logics asks: given a logic $\mathsf{L}$ and a formula $A$, is $A \in \mathsf{L}$, i.e., is $A$ valid in $\mathsf{L}$? Baldoni et al.~\cite[Corollary 2]{BalGioMar98} showed that this problem is undecidable when $\mathcal{C}$ is the class of CGLs. We establish the same result for IGLs by adapting the \emph{negative translation} of Gödel, Gentzen, and Kolmogorov (cf.~\cite{Bus98}) to the modal setting, obtaining a faithful embedding of each CGL into its corresponding IGL. This reduces the general validity problem for CGLs to IGLs, which establishes the undecidability of the problem for IGLs as a corollary.

\subsection{Paper Organization.} Section \ref{sec:log-prelims} introduces IGLs and the grammar theoretic preliminaries needed to formulate our nested sequent systems. Section \ref{sec:nested-calculi} introduces our nested sequent systems and proves them sound and complete. Section \ref{sec:calc-properties} investigates the admissibility and invertibility properties of our systems, and provides a syntactic proof of cut-admissibility that is uniform over all IGLs we consider. Section \ref{sec:interpolation} presents our faithful embedding of CGLs into IGLs via a negative translation, establishing the undecidability of the general validity problem for IGLs. Finally, in Section \ref{sec:conclusion}, we conclude and discuss future work.

\section{Preliminaries}\label{sec:log-prelims}

\subsection{Syntax and Semantics.}


The language of each intuitionistic grammar logic is defined relative to an \emph{alphabet} $\albet$, which is a non-empty finite set of \emph{characters}, used to index modalities. As in~\cite{DemNiv05}, we stipulate that each alphabet $\albet$ is partitioned into a \emph{forward part} $\albetf := \{\chara, \charb, \charc, \ldots\}$ and a \emph{backward part} $\albetb := \{\conv{\chara}, \conv{\charb}, \conv{\charc}, \ldots\}$. That is to say, we assume the following holds for an alphabet $\albet$: (1) $\albet := \albetf \cup \albetb$, (2) $\albetf \cap \albetb = \emptyset$, and (3) $ \chara \in \albetf$ \iffi  $\conv{\chara} \in \albetb$. 

We refer to the elements $\chara$, $\charb$, $\charc$, \etc of $\albetf$ as \emph{forward characters} and the elements $\conv{\chara}$, $\conv{\charb}$, $\conv{\charc}$, \etc of $\albetb$ as \emph{backward characters}. A \emph{character} is defined to be either a forward or backward character, and we use $\charx$, $\chary$, $\charz$, \etc to denote them. In what follows, modalities indexed with forward characters will be interpreted as making reference to future states along the accessibility relation within a relational model, and modalities indexed with backward characters will make reference to past states. 

We define the \emph{converse operation} to be a function $\conv{\cdot}$ mapping each forward character $\chara \in \albetf$ to its \emph{converse} $\conv{\chara} \in \albetb$ and vice versa. Hence, the converse operation is its own inverse, i.e. for any $\charx \in \albet$, $\charx = \conv{\conv{\charx}}$. 

Let us fix an alphabet $\albet$ for the remainder of the text. We let $\prop := \{p, q, r, \ldots\}$ be a denumerable set of \emph{propositional atoms} and define the language $\lang$ relative to a given alphabet $\albet$ via the following grammar in BNF:
$$
A ::= p \mid \bot \mid A \lor A \mid A \land A \mid A \iimp A \mid \xdia A \mid \xbox A
$$
where $p$ ranges over the set $\prop$ of propositional atoms and $\charx$ ranges over the characters in the alphabet $\albet$. We use $A$, $B$, $C$, \etc to denote formulae in $\lang$ and define $\top := \bot \iimp \bot$, $\inot A := A \iimp \bot$, and $A \leftrightarrow B := (A \iimp B) \land (B \iimp A)$ as usual. We define the \emph{length} of a formula $A$, denoted $\ell(A)$, to be the number of symbols contained in $A$.

Formulae are interpreted over \emph{bi-relational $\albet$-frames and models}~\cite{Lyo21b}, which we refer to simply as \emph{frames} and \emph{models}. These frames and models are inspired by those for intuitionistic modal and tense logics~\cite{BovDov84,Dov85,Ewa86,PloSti86}.

\begin{definition}[Frame]\label{def:bi-relational-frame} A \emph{frame} is a tuple $F = (W, \leq, \{R_{\charx} \mid \charx \in \albet\})$ that satisfies the following conditions:
\begin{itemize}

\item $W$ is a non-empty set of \emph{worlds} $\{w, u, v, \ldots\}$;

\item The \emph{intuitionistic relation} $\leq \ \subseteq W \times W$ is a pre-order, i.e., it is reflexive and transitive;

\item The \emph{accessibility relation} $R_{\charx} \subseteq W \times W$ satisfies:

\begin{itemize}

\item[(F1)] For all $w, w', v \in W$, if $w \leq w'$ and $w R_{\charx} v$, then there exists a $v' \in W$ such that $w' R_{\charx} v'$ and $v \leq v'$;

\item[(F2)] For all $w, v, v' \in W$, if $w R_{\charx} v$ and $v \leq v'$, then there exists a $w' \in W$ such that $w \leq w'$ and $w' R_{\charx} v'$;

\item[(F3)] $w R_{\charx} u$ \ifandonlyif $u R_{\conv{\charx}} w$.

\end{itemize}
\end{itemize}
\end{definition}

\begin{definition}[Model]\label{def:bi-relational-model} We define a \emph{model} based on a frame $F$ to be a pair $M = (F, V)$ such that:
$V : W \to 2^{\prop}$ is a \emph{valuation function} satisfying the \emph{monotonicity condition}: (M) for each $w, u \in W$, if $w \leq u$, then $V(w) \subseteq V(u)$.
\end{definition}

\begin{remark} The (F1) condition is implied by the (F2) and (F3) conditions, and the (F2) condition is implied by the (F1) and (F3) conditions. In other words, we need only impose $\{(F1),(F3)\}$ or $\{(F2),(F3)\}$ on a model to characterize our logics. We mention both conditions (F1) and (F2) however since we make explicit use of both conditions later on.
\end{remark}

\begin{figure}[t]\label{fig:f1-f2}

\begin{center}
\begin{tabular}{c @{\hskip 3em} c}
\xymatrix@=2em{
w'\ar@{.>}[rr]|-{R_{\charx}}  & & v' \\
 & (F1) & \\
w\ar[uu]|-{\leq}\ar[rr]|-{R_{\charx}} & & v\ar@{.>}[uu]|-{\leq}
}

&

\xymatrix@=2em{
w'\ar@{.>}[rr]|-{R_{\charx}}  & & v' \\
 & (F2) & \\
w\ar@{.>}[uu]|-{\leq}\ar[rr]|-{R_{\charx}} & & v\ar[uu]|-{\leq}
}
\end{tabular}
\end{center}

\caption{Depictions of the (F1) and (F2) conditions imposed on models. Dotted arrows indicate the relations implied by the presence of the solid arrows.}
\end{figure}

The (F1) and (F2) conditions are depicted in \fig~\ref{fig:f1-f2} and ensure the monotonicity of complex formulae (see \lem~\ref{lem:persistence}) in our models, which is a property characteristic of intuitionistic logics.\footnote{For a discussion of these conditions and related literature, see~\cite[\cptr~3]{Sim94}.} If an accessibility relation $R_{\chara}$ is indexed with a forward character, then we interpret it as a relation to \emph{future} worlds, and if an accessibility relation $R_{\conv{\chara}}$ is indexed with a backward character, then we interpret it as a relation to \emph{past} worlds. Hence, our formulae and related models have a tense character, showing that IGLs generalize the intuitionistic tense logics of Ewald~\cite{Ewa86}.


\begin{definition}[Semantic Clauses]
\label{def:semantic-clauses} Let $M = (W, \leq, \{R_{\charx} \mid \charx \in \albet\}, V)$ be a model with $w \in W$. The \emph{satisfaction relation} $M,w \sat A$ between $w \in W$ of $M$ and a formula $A \in \lang$ is inductively defined as follows:

\begin{itemize}

\item $M,w \sat p$ \ifandonlyif $p \in V(w)$, for $p \in \prop$;

\item $M,w \not\sat \bot$;

\item $M,w \sat A \lor B$ \ifandonlyif $M,w \sat A$ or $M,w \sat B$;

\item $M,w \sat A \land B$ \ifandonlyif $M,w \sat A$ and $M,w \sat B$;


\item $M,w \sat A \iimp B$ \ifandonlyif for all $u \in W$, if $w \leq u$ and $M,u \sat A$, then $M,u \sat B$;

\item $M,w \sat \xdia A$ \ifandonlyif there exists a $u \in W$ such that $w R_{\charx} u$ and $M,u \sat A$;

\item $M,w \sat \xbox A$ \ifandonlyif for all $u, v \in W$, if $w \leq u$ and $u R_{\charx} v$, then $M,v \sat A$;

\item $M \sat A$ \iffi for all $w \in W$, $M, w \sat A$.

\end{itemize}
A formula $A$ is \emph{valid} relative to a class of frames $\frames$ \iffi for all models $M$ based on a frame $F \in \frames$, $M \sat A$.
\end{definition}

\begin{lemma}[General Monotonicity]\label{lem:persistence}
Let $M$ be a model with $w,u \in W$ of $M$. If $w \leq u$ and $M, w \sat A$, then $M, u \sat A$.
\end{lemma}

\begin{proof}
By induction on the length of $A$.
\end{proof}

\subsection{Intuitionistic Grammar Logics.}

The base logic $\ikm$ is defined 
by means of the axiom system $\h\ikm$ given in \dfn~\ref{def:axiomatization} below. We remark that we refer to the axiom system as $\h\ikm$, using the prefix $\h$ to denote the fact that the axiom system is a \emph{Hilbert calculus}.

\begin{definition}[Axiom System $\h\ikm$]\label{def:axiomatization} We define the axiom system $\h\ikm$ below, where we have an axiom and inference rule for each $\charx \in \albet$.
\begin{multicols}{2}
\begin{itemize}

\item[A0] Axioms for intuitionistic logic 

\item[A1] $\xbox (A \iimp B) \iimp (\xbox A \iimp \xbox B)$

\item[A2] $\xbox (A \land B) \ieq (\xbox A \land \xbox B)$

\item[A3] $\xdia (A \lor B) \ieq (\xdia A \lor \xdia B)$

\item[A4] $\xbox (A \iimp B) \iimp (\xdia A \iimp \xdia B)$

\item[A5] $(\xbox A \land \xdia B) \iimp \xdia (A \land B)$

\item[A6] $\inot \xdia \bot$ 

\item[A7] $(A \iimp \xbox \xdiac A) \land (\xdia \xboxc A \iimp A)$

\item[A8] $(\xdia A \iimp \xbox B) \iimp \xbox (A \iimp B)$

\item[A9] $\xdia (A \iimp B) \iimp (\xbox A \iimp \xdia B)$

\item[R0] \AxiomC{$A$}\AxiomC{$A \iimp B$}\RightLabel{$\mathsf{mp}$}\BinaryInfC{$B$}\DisplayProof

\item[R1] \AxiomC{$A$}\RightLabel{$\nec$}\UnaryInfC{$\xbox A$}\DisplayProof

\end{itemize}
\end{multicols}
\end{definition}

\begin{remark} We note that if we let $\albet := \{a,\conv{a}\}$, then $\ikm$ is just Ewald's intuitionistic tense logic $\ikt$~\cite{Ewa86}, which is a conservative extension of the mono-modal intuitionistic modal logic $\ik$~\cite{Fis84,PloSti86}. 
\end{remark}

In this paper, we also consider extensions of $\ikm$ 
with two different classes of axioms, where $\charx, \charx_{1}, \ldots, \charx_{n} \in \albet$: 
\begin{itemize}

\item \emph{seriality axioms} $\xbox A \iimp \xdia A$ 

\item \textit{intuitionistic path axioms} $(\diap{\charx_{1}} \cdots \diap{\charx_{n}} A \iimp \diap{\charx} A) \land (\boxp{\charx} A \iimp \boxp{\charx_{1}} \cdots \boxp{\charx_{n}} A)$

\end{itemize}
We let $\axs$ denote a finite set of axioms of the above forms and refer to intuitionistic path axioms as \emph{IPAs}. 
We note that the collection of $\ipa$s subsumes the class of Horn-Scott-Lemmon axioms~\cite{Lyo21a} and includes multi-modal and intuitionistic variants of standard axioms such as $\T$, $\B$, $\four$, and $\five$:
\begin{itemize}

\item[$\T$] $(A \iimp \xdia A) \land (\xbox A \iimp A)$

\item[$\four$] $(\xdia \xdia A \iimp \xdia A) \land (\xbox A \iimp \xbox \xbox A)$

\item[$\B$] $(\xdiac A \iimp \xdia A) \land (\xbox A \iimp \xboxc A)$

\item[$\five$] $(\xdiac \xdia A \iimp \xdia A) \land (\xbox A \iimp \xboxc \xbox A)$

\end{itemize}

\begin{definition}[$\h\ikm(\axs)$] We define $\h\ikm(\axs) := \h\ikm \cup \axs$, and define the \emph{intuitionistic grammar logic} $\ikm(\axs)$ to be the smallest set of formulae from $\lang$ closed under substitutions of the axioms from $\h\ikm(\axs)$ and applications of the inference rules. A formula $A$ is defined to be a \emph{theorem} of $\ikm(\axs)$ \iffi $A \in \ikm(\axs)$. 
Given a set $\Gamma \subseteq \lang$ and formula $A$, we define $\Gamma \ent A$ \iffi there exist $B_{1}, \ldots, B_{n} \in \Gamma$ such that $B_{1} \land \cdots \land B_{n} \iimp A$ is a theorem of $\ikm(\axs)$.
\end{definition}

\begin{remark} The axiom system $\h\ikm = \h\ikm(\emptyset)$. 
\end{remark}

As stated above and proven in~\cite{Lyo21b}, each logic $\ikm(\axs)$ admits a semantic characterization in terms of classes of frames, determined by imposing frame conditions which take one of two different forms:
\begin{itemize}

\item \emph{seriality conditions} $\forall w \exists u (w R_{\charx} u)$

\item \emph{path conditions} $\forall \vec{w} (w_{0} R_{\charx_{1}} w_{1} \& \cdots \& w_{n-1} R_{\charx_{n}} w_{n} \Rightarrow w_{0} R_{\charx} w_{n})$

\end{itemize}
We let $\vec{w} = w_{0}, \ldots, w_{n}$ in the path condition above. We say that each aforementioned seriality axiom is \emph{related} to the seriality condition of the above form, and each IPA is \emph{related} to a path condition of the above form. 

\begin{definition}[$\axs$-Valid] We define an \emph{$\axs$-frame} to be a frame satisfying each frame condition related to an axiom $A \in \axs$. An \emph{$\axs$-model} is a model based on an $\axs$-frame. A formula $A$ is defined to be \emph{$\axs$-valid}, written $\asat A$, \iffi $A$ is valid relative to the class of $\axs$-frames. For a set $\Gamma \subseteq \lang$, we define $\Gamma \asat A$ \iffi for every $\axs$-model $M$ and for all worlds $w \in W$ of $M$, if $M, w \sat B$ for all $B \in \Gamma$, then $M, w \sat A$.
\end{definition}

The following soundness and completeness result is proven in~\cite{Lyo21b}.


\begin{theorem}[Soundness and Completeness]\label{thm:sound-complete-logic} 
$\Gamma \ent A$ \iffi $\Gamma \asat A$.
\end{theorem}

\subsection{Grammar-Theoretic Preliminaries.} 

As will be seen later on, the definition of our nested sequent systems crucially depends upon the use of certain inference rules (viz. propagation rules) parameterized by formal grammars. Therefore, the current section introduces the grammar-theoretic notions required to properly define such rules.

We let $\concat$ be the \emph{concatenation operation} with $\varepsilon$ the \emph{empty string}. We define the set $\albet^{*}$ of \emph{strings over $\albet$} to be the smallest set such that (1) $\albet \cup \{\varepsilon\} \subseteq \albet^{*}$, and (2) $\text{If } \stra \in \albet^{*} \text{ and } \charx \in \albet \text{, then } \stra \concat \charx \in \albet^{*}$. We use $\stra$, $\strb$, $\strc$, \etc (potentially annotated) to represent strings in $\albetstr$. Also, we will often write the concatenation of two strings $\stra$ and $\strb$ as $\stra \cate \strb$ as opposed to $\stra \concat \strb$, and for the empty string we have $\stra \cate \empstr = \empstr \cate \stra = \stra$. The converse operation on strings is defined accordingly: (1)  $\conv{\varepsilon} := \varepsilon$, and (2) $\text{If } \stra = \charx_{1} \cdots \charx_{n} \text{, then } \conv{\stra} := \conv{\charx}_{n} \cdots \conv{\charx}_{1}$.











To simplify the presentation of formulae in the sequel, we define strings in modalities as follows: if $\stra = \charx_{1} \cdots \charx_{n}$, then $[ \stra ] A :=  [ \charx_{1} ] \cdots [ \charx_{n} ] A$ and $\langle \stra \rangle A := \langle \charx_{1} \rangle \cdots \langle \charx_{n} \rangle A$, with $\strabox A = \stradia A = A$ when $\stra = \empstr$. Hence, every $\ipa$ may be written in the form
$$
(\langle \stra \rangle A \iimp \langle \charx \rangle A) \land ([ \charx ] A \iimp [\stra] A)
$$
where $\langle \stra \rangle = \langle \charx_{1} \rangle \cdots \langle \charx_{n} \rangle$ and $[ \stra ] = [ \charx_{1} ] \cdots [ \charx_{n} ]$. We make use of this notation to compactly define \emph{$\axs$-grammars}, which are types of \emph{Semi-Thue systems}~\cite{Pos47}, encoding information contained in a set $\axs$ of axioms, and employed later on in the definition of propagation rules. 

\begin{definition}[$\axs$-grammar]\label{def:grammar} An \emph{$\axs$-grammar} is a set $\g{\axs}$ such that:
\begin{center}
$(\charx \pto \stra), (\conv{\charx} \pto \conv{\stra}) \in \g{\axs}$ \iffi $(\langle \stra \rangle A \iimp \langle \charx \rangle A) \land ([ \charx ] A \iimp [\stra] A) \in \axs$.
\end{center}
We call rules of the form $\charx \pto \stra$ \emph{production rules}, where $\charx \in \albet$ and $\stra \in \albetstr$.
\end{definition}


An $\axs$-grammar $\g{\axs}$ is a type of string re-writing system. For example, if $\charx \pto \stra \in \g{\axs}$, we may derive the string $\strb \cate \stra \cate \strc$ from $\strb \cate \charx \cate \strc$ in one-step by applying the mentioned production rule. Through repeated applications of production rules to a given string $\stra \in \albetstr$, one derives new strings, the collection of which, determines a language. Let us make such notions precise by means of the following definition:

\begin{definition}[Derivation]\label{def:gram-derivation} Let $\g{\axs}$ be an $\axs$-grammar. The \emph{one-step derivation relation} $\osdr$ holds between two strings $\stra$ and $\strb$ in $\albetstr$, written $\stra \osdr \strb$, \iffi there exist $\stra', \strb' \in \albetstr$ and $\charx \pto \strc \in \g{\axs}$ such that $\stra = \stra' \cate \charx \cate \strb'$ and $\strb = \stra' \cate \strc \cate \strb'$.  The \emph{derivation relation} $\dr$ is defined to be the reflexive and transitive closure of $\osdr$. For two strings $\stra, \strb \in \albetstr$, we refer to $\stra \dr \strb$ as a \emph{derivation of $\strb$ from $\stra$}, and define its \emph{length} to be equal to the minimal number of one-step derivations needed to derive $\strb$ from $\stra$ in $\g{\axs}$. 
\end{definition}

\begin{definition}[Language]\label{def:gram-language} Let $\g{\axs}$ be an $\axs$-grammar. For a character $\charx \in \albet$, the \emph{language of $\charx$ relative to $\g{\axs}$} is defined to be the set $\glang(\charx) := \{\stra \ | \ \charx \dr \stra \}$. When the $\axs$-grammar $\g{\axs}$ is clear from the context, we write $L_{\charx}$ to denote $\glang(\charx)$.
\end{definition}

The following lemma will be helpful in the sequel. It follows from the definition of $\g{\axs}$, namely, because for each production rule $\charx \pto \stra \in \g{\axs}$ there is a production rule $\conv{\charx} \pto \conv{\stra} \in \g{\axs}$.

\begin{lemma}\label{lem:string-and-inverse} 
$\stra = \charx_{1} \cdots \charx_{n} \in \glang(\charx)$ \iffi \ $\conv{\stra} = \conv{\charx}_{n} \cdots \conv{\charx}_{1} \in \glang(\conv{\charx})$. 
\end{lemma}



\section{Nested Sequent Systems}\label{sec:nested-calculi}


Let $\Gamma$, $\Delta$, $\Sigma$, $\ldots$ be finite multisets of formulae from $\lang$. We define a \emph{Gentzen sequent} to be an expression of the form $\Gamma \sar \Delta$ such that $|\Delta| \leq 1$. We define a \emph{nested sequent} to be a \emph{right-empty} or \emph{right-filled} nested sequent, which are simultaneously defined as follows:
\begin{enumerate}

\item[(1)] A Gentzen sequent $\Gamma \sar \Delta$ with $|\Delta| = 0$ is a right-empty nested sequent, and a Gentzen sequent $\Gamma \sar \Delta$ with $|\Delta| = 1$ is a right-filled nested sequent;

\item[(2)] if $\Gamma \sar \Delta$ is right-empty (or, right-filled) and $\ns_{i}$ is a right-empty nested sequent for each $i \in [n]$, then $\Gamma \sar \Delta, (\charx_{1})[\ns_{1}], \ldots, (\charx_{n})[\ns_{n}]$ is a right-empty (right-filled, resp.) nested sequent;

\item[(3)] if $\Gamma \sar \Delta$ is right-empty, $\ns_{j}$ is a right-filled nested sequent for a unique $j \in [n]$, and $\ns_{i}$ is a right-empty nested sequent for each $i \in [n] \setminus \set{j}$, then $\Gamma \sar \Delta, (\charx_{1})[\ns_{1}], \ldots, (\charx_{n})[\ns_{n}]$ is a right-filled nested sequent.

\end{enumerate}
We use the symbols $\ns$, $\nsii$, $\nsiii$, $\ldots$ to denote nested sequents. For 
a nested sequent of the form $\Gamma \sar \Delta, (\charx_{1})[\ns_{1}], \ldots, (\charx_{n})[\ns_{n}]$, we refer to $\Gamma$ as the \emph{antecedent} and refer to $\Delta, (\charx_{1})[\ns_{1}], \ldots, (\charx_{n})[\ns_{n}]$ as the \emph{consequent}. 

We define a \emph{component} of the nested sequent $\ns =  \Gamma \sar \Delta, (\charx_{1})[\nsii_{1}], \ldots, (\charx_{n})[\nsii_{n}]$ to be a Gentzen sequent appearing in $\ns$, that is, a component of $\ns$ is an element of the multiset $c(\ns)$, where $c$ is defined as follows and $\uplus$ denotes multiset union:
$$
c( \Gamma \sar \Delta, (\charx_{1})[\nsii_{1}], \ldots, (\charx_{n})[\nsii_{n}]) := \{ \Gamma \sar \Delta \} \uplus \! \biguplus_{i \in [n]} \! c(\nsii_{i}).
$$
We let $\names := \set{w,u,v,\ldots}$ be a denumerable set of pairwise distinct labels, called \emph{names}. For each nested sequent $\ns$ of the form $\Gamma \sar \Delta, (\charx_{1})[\nsii_{1}], \ldots, (\charx_{n})[\nsii_{n}]$, we call $\Gamma \sar \Delta$ the \emph{root} of $\ns$ and assume that every component of $\ns$ is assigned a unique name from $\names$. The incorporation of names in our nested sequents is crucial for the definition of our propagation rules below. We define $\names(\ns)$ to be the set of all names assigned to components in $\ns$.

We define the \emph{input formulae} and \emph{output formulae} of a nested sequent $\ns$ to be all formulae occurring in antecedents of components (shown below left) and consequents of components (shown below right), respectively:
$$
\infs{\ns} := \!\!\!\!\!\!\!\! \bigcup_{(\Gamma \sar \Delta) \in c(\ns)} \!\!\!\!\!\!\!\! \Gamma
\qquad
\outfs{\ns} := \!\!\!\!\!\!\!\! \bigcup_{(\Gamma \sar \Delta) \in c(\ns)} \!\!\!\!\!\!\!\! \Delta
$$
Observe that $\outfs{\ns} = \emptyset$ if $\ns$ is right-empty, and $\outfs{\ns} = \set{A}$ if $\ns$ is right-filled; in the latter case, we simply write $\outfs{\ns} = A$. 
For a name $w$ of a component $\Gamma \sar \Delta$ in a nested sequent $\ns$, we define $\infs{w,\ns} := \Gamma$ and $\outfs{w,\ns} := \Delta$. 
Following the notation of~\cite{KuzStr19,Str13}, we define $\ns^{\da}$ to be the nested sequent $\ns$ with the output formula (if it exists) removed.

Observe that every nested sequent encodes a \emph{tree} of (named) Gentzen sequents.

\begin{definition} Let $\ns =  \Gamma \sar \Delta, (\charx_{1})[\nsii_{1}], \ldots, (\charx_{n})[\nsii_{n}]$ be a nested sequent with $w$ the name of its root and $u_{i}$ the name of the root of $\nsii_{i}$ for $i \in [n]$. We recursively define the tree $\tr(\ns) = (T,E)$ as follows:
$$
T = \{(w,\Gamma \sar \Delta)\} \cup \bigcup_{i = 1}^{n} T_{i}
\qquad
E = \{(w,\charx_{i},u_{i}) \ | \ i \in [n]\} \cup \bigcup_{i = 1}^{n} E_{i}
$$
such that $\tr(\nsii_{i}) = (T_{i},E_{i})$ for $i \in [n]$. We will often write $(u,\Sigma \sar \Pi) \in \tr(\ns)$ and $(u,\charx,v) \in \tr(\ns)$ to mean that $(u,\Sigma\sar \Pi) \in T$ and $(u,\charx,v) \in E$, respectively.
\end{definition}

We sometimes refer to a component as a \emph{$w$-component} if $w$ is the name of that component, and we use the notation $\ns\{\nsii_{1}\}_{w_{1}}\ldots\{\nsii_{n}\}_{w_{n}}$ to indicate that $\nsii_{i}$ is `rooted at' the $w_{i}$-component of $\ns$ for $i \in [n]$. Often, we will more simply write $\ns\{\nsii_{1}\}\ldots\{\nsii_{n}\}$ if the names of the components are not of relevance. To reduce notational complexity, we may \emph{not} write $\nsii_{i}$ 
in its entirety in $\ns\{\nsii_{1}\}\ldots\{\nsii_{n}\}$, but may only display the root and/or some children of the root. For example, $\ns\{A, B \sar \emptyset\}$, $\ns\set{C, D \sar \emptyset}\set{I \sar J}$, and $\ns\set{C, D \sar \emptyset, (\chary)[I \sar J]}$ are all acceptable representations of the nested sequent $\ns = A, B \sar \emptyset, (\charx)[C, D \sar \emptyset, (\chary)[I \sar J]]$.



\begin{definition}[Sequent Semantics]\label{def:sequent-semantics} Let $\ns$ be a nested sequent with $w$ the name of the root, $\tr(\ns) = (T,E)$, and $M = (W, \leq, \{R_{\charx} \ | \ \charx \in \albet\}, V)$ be an $\axs$-model. We define an \emph{$M$-interpretation} to be a function $\mint : \names \rightarrow W$. We say that $\ns$ is satisfied on $M$ with $\mint$, written $M, \mint \asat \ns$, \iffi if condition (1) holds, then (2) holds, where conditions (1) and (2) are listed below:
\begin{itemize}

\item[(1)] for each $(u,\charx,v) \in E$, $\mint(u)R_{\charx}\mint(v)$;

\item[(2)] for some $(v,\Gamma \sar \Delta) \in T$, $M, \mint(v) \asat \bigwedge \Gamma \iimp \bigvee \Delta$. 

\end{itemize}
We say that $\ns$ is \emph{$\axs$-valid} \iffi for every $\axs$-model $M$ and $M$-interpretation $\mint$, $M, \mint \asat \ns$; we say that $\ns$ is \emph{$\axs$-invalid} otherwise.
\end{definition}



\begin{figure}[t]
\noindent

\begin{center}
\begin{tabular}{c c} 
\AxiomC{}
\RightLabel{$\id$}
\UnaryInfC{$\ns \sbl \Gamma, p \sar p \sbr$}
\DisplayProof

&

\AxiomC{}
\RightLabel{$\botl$}
\UnaryInfC{$\ns \sbl \Gamma, \bot \sar \Delta \sbr$}
\DisplayProof
\end{tabular}
\end{center}


\begin{center}
\begin{tabular}{c c c}
\AxiomC{$\ns \sbl \Gamma, A \sar \Delta \sbr$}
\AxiomC{$\ns \sbl \Gamma, B \sar \Delta \sbr$}
\RightLabel{$\disl$}
\BinaryInfC{$\ns \sbl \Gamma, A \lor B \sar \Delta \sbr$}
\DisplayProof

&

\AxiomC{$\ns \sbl \Gamma \sar A_{i} \sbr$}
\RightLabel{$\disr~\text{s.t.}~i \in \set{1,2}$}
\UnaryInfC{$\ns \sbl \Gamma \sar A_{1} \lor A_{2} \sbr$}
\DisplayProof
\end{tabular}
\end{center}


\begin{center}
\begin{tabular}{c c c}
\AxiomC{$\ns \sbl \Gamma, A, B \sar \Delta \sbr$}
\RightLabel{$\conl$}
\UnaryInfC{$\ns \sbl \Gamma, A \land B \sar \Delta \sbr$}
\DisplayProof

&

\AxiomC{$\ns \sbl \Gamma \sar A \sbr$}
\AxiomC{$\ns \sbl \Gamma \sar B \sbr$}
\RightLabel{$\conr$}
\BinaryInfC{$\ns \sbl \Gamma \sar A \land B \sbr$}
\DisplayProof
\end{tabular}
\end{center}


\begin{center}
\begin{tabular}{c c}
\AxiomC{$\ns^{\downarrow} \sbl \Gamma, A \iimp B \sar A \sbr$}
\AxiomC{$\ns \sbl \Gamma, B \sar \Delta \sbr$}
\RightLabel{$\iimpl$}
\BinaryInfC{$\ns \sbl \Gamma, A \iimp B \sar \Delta \sbr$}
\DisplayProof

&

\AxiomC{$\ns\sbl \Gamma, A \sar B \sbr$}
\RightLabel{$\iimpr$}
\UnaryInfC{$\ns \sbl \Gamma \sar A \iimp B \sbr$}
\DisplayProof
\end{tabular}
\end{center}


\begin{center}
\begin{tabular}{c c}
\AxiomC{$\ns \sbl \Gamma \sar \emptyset \sbr_{w} \sbl \Sigma \sar A \sbr_{u}$}
\RightLabel{$\xdiar^{\dag(\axs)}$}
\UnaryInfC{$\ns \sbl \Gamma \sar \xdia A \sbr_{w} \sbl \Sigma \sar \emptyset \sbr_{u}$}
\DisplayProof

&

\AxiomC{$\ns \sbl \Gamma, \xbox A \sar \Delta \sbr_{w} \sbl \Sigma, A \sar \Pi \sbr_{u}$}
\RightLabel{$\xboxl^{\dag(\axs)}$}
\UnaryInfC{$\ns \sbl \Gamma, \xbox A \sar \Delta \sbr_{w} \sbl \Sigma \sar \Pi \sbr_{u}$}
\DisplayProof
\end{tabular}
\end{center}


\begin{center}
\begin{tabular}{c @{\hskip .1cm} c @{\hskip .1cm} c}
\AxiomC{$\ns \sbl \Gamma \sar \Delta, (\charx) [ A \sar \emptyset ] \sbr$}
\RightLabel{$\xdial$}
\UnaryInfC{$\ns \sbl \Gamma, \xdia A \sar \Delta \sbr$}
\DisplayProof

&

\AxiomC{$\ns \sbl \Gamma \sar (\charx) [ \emptyset \sar A ] \sbr$}
\RightLabel{$\xboxr$}
\UnaryInfC{$\ns \sbl \Gamma \sar \xbox A \sbr$}
\DisplayProof

&

\AxiomC{$\ns \sbl \Gamma \sar \Delta, (\charx)\bl \empseq \sar \empseq \br \sbr$}
\RightLabel{$\ddr$}
\UnaryInfC{$\ns \sbl \Gamma \sar \Delta \sbr$}
\DisplayProof
\end{tabular}
\end{center}


\begin{flushleft}
\textbf{Side Conditions:}\\
$\dag(\axs):= w \prpath{L_{\charx}} u$ with $L_{\charx} = \glang(\charx)$.
\end{flushleft}

\caption{The nested sequent system $\calc$ defined relative to the set $\axs$ of axioms. Note that $\ddr \in \calc$ \iffi $\D \in \axs$.\label{fig:nested-calculus}}
\end{figure}

A uniform presentation of our nested sequent systems is provided in \fig~\ref{fig:nested-calculus}. Each nested sequent system $\calc$ takes a set $\axs$ of axioms as a parameter and is sound and complete for the logic $\ikm(\axs)$ (see Theorems~\ref{thm:soundness-nested} and~\ref{thm:completeness-nested} below). We note that  $\mathsf{NIK_{m}} = \mathsf{NIK_{m}}(\emptyset)$ is the calculus for the base logic $\ikm$. Each system contains the \emph{initial rules} $\id$ and $\botl$, and with the exception of $\ddr$, the remaining rules are \emph{logical rules} forming left and right pairs. 

As stated above, the parameter $\axs$ determines what logic is captured by the system $\calc$. This is achieved by (1) adding the \emph{structural rule} $\ddr$ to $\calc$ \iffi $\D \in \axs$ and/or (2) changing the functionality of the propagation rules $\xdiar$ and $\xboxl$. Propagation rules are special in that they view nested sequents as automata, and enable formulae to be (bottom-up) propagated along certain paths (corresponding to strings generated by a $\albet$-system) within the tree $\tr(\ns)$ of a nested sequent $\ns$. To make the functionality of such rules precise, we define \emph{propagation graphs} and \emph{propagation paths} (cf.~\cite{CiaLyoRamTiu21,GorPosTiu11}).

\begin{definition}[Propagation Graph]\label{def:propagation-graph} Let $\ns$ be a nested sequent with $\tr(\ns) = (T,E)$. We define the \emph{propagation graph} $\prgr{\ns} = (\prv,\pre)$ such that 

\begin{itemize}

\item[(1)] $\prv := \names(\ns)$;

\item[(2)] $(w, \charx, u), (u, \conv{\charx}, w) \in \pre$ \iffi $(w,\charx,u) \in E$.


\end{itemize}
We will often write $w \in \prgr{\ns}$ and $(w,\charx,u) \in \prgr{\ns}$ as a shorthand for $w \in \prv$ and $(w,\charx,u) \in \pre$, respectively.
\end{definition}

\begin{definition}[Propagation Path]\label{def:propagation-path} Let $\prgr{\ns} = (\prv,\pre)$ be a propagation graph with $u,w \in \prv$ and $\charx \in \albet$. We write $\prgr{\ns} \models u \prpath{\charx} w$ \iffi $(u,\charx,w) \in \pre$. Given a string $\charx\stra \in \albetstr$, 
we define $\prgr{\ns} \models u \prpath{\charx\stra} w$ as `$\exists_{v \in \prv} \ \prgr{\ns}  \models u\prpath{\charx} v$ and $\prgr{\ns} \models v \prpath{\stra}w$', and we take $\prgr{\ns} \models u \prpath{\empstr} w$ to mean that $u = w$. Additionally, when $\prgr{\ns}$ is clear from the context we may simply write $u \prpath{\stra} w$ to express $\prgr{\ns} \models u \prpath{\stra} w$. Finally, given a language $L_{\charx} := \glang(\charx)$ of some $\axs$-system $\g{\axs}$ and $\stra \in \albetstr$, we write $u \prpath{L_{\charx}} w$ \iffi there is a string $\strb \in \glang(\charx)$ with $u \prpath{\strb} w$.
\end{definition}

The following is a consequence of \lem~\ref{lem:string-and-inverse}.

\begin{lemma}\label{lem:reversing-paths}
Let $\ns$ be a nested sequent, $L_{\charx} = \glang(\charx)$, and $L_{\conv{\charx}} = \glang(\conv{\charx})$. Then, $\prgr{\ns} \models u \prpath{L_{\charx}} w$ \iffi $\prgr{\ns} \models w \prpath{L_{\conv{\charx}}} u$.
\end{lemma}

With the above notions, one can formally specify the propagation rules $\xdiar$ and $\xboxl$. Applications of such rules are controlled by the side condition $\dag(\axs)$ shown in \fig~\ref{fig:nested-calculus}. We provide an example to clarify applications of these rules.

\begin{example}\label{ex:propagation-graph-path} Let $\ns = \Gamma, p \sar \Delta, (\charx)[\Sigma \sar \Pi, (\conv{\chary})[\Phi, [z]p \sar \Psi]]$ with $w$ the name of the root, $u$ the name of $\Sigma \sar \Pi$, and $v$ the name of $\Phi, [z]p \sar \Psi$. A graphical depiction of $\ns$'s propagation graph $\prgr{\ns}$ is given below:
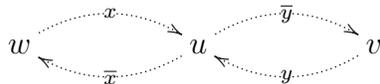
\begin{center}
\begin{minipage}[t]{.5\textwidth}
\xymatrix{
  w \ar@/^1pc/@{.>}[rr]|-{\charx} & &  u\ar@/^1pc/@{.>}[rr]|-{\conv{\chary}} \ar@/^1pc/@{.>}[ll]|-{\conv{\charx}} & &  v \ar@/^1pc/@{.>}[ll]|-{\chary}
}
\end{minipage}
\end{center}
Let $\axs := \{(\langle \chary \rangle \langle \conv{\charx} \rangle A \iimp \langle z \rangle A) \land ([z] A \iimp [\chary] [\conv{\charx}] A)\}$, so that the corresponding $\axs$-grammar is $\g{\axs} = \{z \pto \chary \conv{\charx}, \conv{z} \pto \conv{\chary} \charx \}$. Then, $\prgr{\ns} \models v \prpath{y\conv{x}} w$ with $\chary \conv{\charx} \in \glang(z)$ due to the first production rule of $\g{\axs}$. Therefore, the side condition of $\xboxl$ is satisfied, and so, we are permitted to (top-down) apply the propagation rule $\xboxl$ to $\ns$ to delete the formula $p$. This lets us derive the nested sequent $\Gamma \sar \Delta, (\charx)[\Sigma \sar \Pi, (\conv{\chary})[\Phi, [z]p \sar \Psi]]$.
\end{example}

As usual, we refer to formulae that are explicitly presented in the premises and conclusion of a rule as \emph{auxiliary} and \emph{principal}, respectively. For example, in the $\xboxr$ rule, $A$ is auxiliary and $\xbox A$ is principal. A \emph{derivation} in $\calc$ of a nested sequent $\ns$ is a (potentially infinite) tree whose nodes are labeled with nested sequents such that: (1) The root is labeled with $\ns$, and (2) Every parent node is the conclusion of a rule of $\calc$ with its children the premises.

We define a \emph{branch} $\branch = \ns_{0}, \ns_{1}, \ldots, \ns_{n}, \ldots$ to be a maximal path of nested sequents in a derivation such that $\ns_{0}$ is the conclusion of the derivation and each nested sequent $\ns_{i+1}$ (if it exists) is a child of $\ns_{i}$. A \emph{proof} is a finite derivation where all leaves are instances of initial rules. We use $\prf$ and annotated versions thereof to denote both derivations and proofs with the context differentiating the usage. If a proof of a nested sequent $\ns$ exists in $\calc$, then we write $\calc\proves\ns$ to indicate this. Often, when discussing a proof transformation, we will use the following notation to denote that $\prf$ is a proof of $\ns$: 
\begin{center}
\AxiomC{$\dotprf{\prf}$}
\UnaryInfC{$\ns$}
\DisplayProof
\end{center}
The \emph{height} of a derivation is defined in the usual way as the length of a maximal branch in the derivation, which may be infinite if the derivation is infinite.

\begin{example} We provide an example proof of an IPA below. We suppose that the IPA shown below is an axiom in $\axs$, meaning $ \charx_{1}\cdots\charx_{n} \pto \charx \in \g{\axs}$ by definition. Observe that a propagation path tracing the word $\charx_{1}\cdots\charx_{n}$ exists in the left and right branches, showing that the side condition holds for $\xdiar$ and $\xboxl$ as $\charx_{1}\cdots\charx_{n} \in \glang(x)$.
\begin{center}
\resizebox{\textwidth}{!}{
\AxiomC{}
\RightLabel{$\id$}

\UnaryInfC{$\emptyset \sar (\charx_{1})[\ldots (\charx_{n})[p \sar p] \ldots]$}
\RightLabel{$\xdiar$}
\UnaryInfC{$\emptyset \sar \diap{\charx} p, (\charx_{1})[\ldots (\charx_{n})[p \sar \emptyset] \ldots]$}
\RightLabel{$\xdial \times n$}
\UnaryInfC{$\diap{\charx_{1}} \cdots \diap{\charx_{n}} p \sar \diap{\charx} p$}
\RightLabel{$\iimpr$}
\UnaryInfC{$\sar \diap{\charx_{1}} \cdots \diap{\charx_{n}} p \iimp \diap{\charx} p$}

\AxiomC{}
\RightLabel{$\id$}

\UnaryInfC{$\boxp{\charx} p \sar (\charx_{1})[\ldots (\charx_{n})[p \sar p] \ldots]$}
\RightLabel{$\xboxl$}
\UnaryInfC{$\boxp{\charx} p \sar (\charx_{1})[\ldots (\charx_{n})[\emptyset \sar p] \ldots]$}
\RightLabel{$\xboxr \times n$}
\UnaryInfC{$\boxp{\charx} p \sar \boxp{\charx_{1}} \cdots \boxp{\charx_{n}} p$}
\RightLabel{$\iimpr$}
\UnaryInfC{$\sar \boxp{\charx} p \iimp \boxp{\charx_{1}} \cdots \boxp{\charx_{n}} p$}

\RightLabel{$\conr$}
\BinaryInfC{$\sar (\diap{\charx_{1}} \cdots \diap{\charx_{n}} p \iimp \diap{\charx} p) \land (\boxp{\charx} p \iimp \boxp{\charx_{1}} \cdots \boxp{\charx_{n}} p)$}
\DisplayProof
}
\end{center}
\end{example}

\subsection{Soundness and Completeness}

We first prove the soundness of our nested systems, and afterward, discuss the proof of completeness. Given a model $M$ and string $\stra = \charx_{1} \cdots \charx_{n}$, we define the relation $R_{\stra} := R_{\charx_{1}} \circ \cdots \circ R_{\charx_{n}}$.

\begin{lemma}\label{lem:paths-implies-edge-in-model}
Let $M = (W, \leq, \{R_{\charx} \ | \ \charx \in \albet\},V)$ be an $\axs$-model. If 
$wR_{\stra}u$ and $\stra \in L_{\charx}$, then $w R_{\charx} u$. 
\end{lemma}

\begin{proof} By induction on the length of the derivation of $\stra \in \glang(\charx)$. 

\textit{Base case.} For the base case, suppose the length of the derivation of $\stra$ is $0$. By \dfn~\ref{def:gram-language}, the only derivation in $\glang(\charx)$ of length $0$ is the derivation consisting solely of $\charx$. Hence, the claim trivially holds since $\stra = \charx$.

\textit{Inductive step.} Let the derivation of $\stra$ be of length $n+1$. By \dfn~\ref{def:gram-language}, there is a derivation $\charx \pto^{*}_{\g{\axs}} \strb \chary \strc$ of length $n$ and there exists a production rule $\chary \pto \charz_{1} \cdots \charz_{k} \in \g{\axs}$ such that $\stra = \strb \charz_{1} \cdots \charz_{k} \strc$. By our assumption that $wR_{\stra}u$, we know there exist $v, v' \in W$ such that $w R_{\strb} v$, $v R_{\charz_{1} \cdots \charz_{k}} v'$, and $v' R_{\strc} u$. By the definition of $\g{\axs}$, it follows that $v R_{\chary} v'$, and so, we have $w R_{\strb} v$, $v R_{\chary} v'$, and $v' R_{\strc} u$, i.e., $w R_{\strb \chary \strc} u$. Therefore, $w R_{\charx} u$ by IH.
\end{proof}

\begin{lemma}\label{lem:paths-imply-relation}
Let $M = (W, \leq, \{R_{\charx} \ | \ \charx \in \albet\},V)$ be an $\axs$-model and let $\mint$ be an $M$-interpretation. If $M, \mint \not\asat \ns$ and $\prgr{\ns} \models w \prpath{L_{\charx}} u$, then $\mint(w) R_{\charx} \mint(u)$.
\end{lemma}

\begin{proof} By assumption, there is a string $\stra \in L_{\charx}$ such that $\prgr{\ns} \models w \prpath{\stra} u$. By \dfn~\ref{def:sequent-semantics}, $\mint(w) R_{\stra} \mint(u)$, and so, $\mint(w) R_{\charx} \mint(u)$ by \lem~\ref{lem:paths-implies-edge-in-model}.
\end{proof}

\begin{theorem}[Soundness]\label{thm:soundness-nested} If $\calc \proves \ns$, then $\ns$ is $\axs$-valid.
\end{theorem}

\begin{proof} It is straightforward to show that all instances of $\id$ and $\botl$ are valid. Therefore, let us argue soundness by showing that the other rules of $\calc$ are sound, i.e., if the conclusion is $\axs$-invalid, then at least one premise is $\axs$-invalid. We present the $\iimpr$ and $\xdiar$ cases; the remaining cases are simple or similar. Let $M = (W, \leq, \{R_{\charx} \ | \ \charx \in \albet\},V)$ be an $\axs$-model with $\mint$ an $M$-interpretation.


$\iimpr$. Let $\nsii := \ns \sbl \Gamma \sar A \iimp B \sbr_{w}$ be the conclusion of an $\iimpr$ instance and suppose $M, \mint \not\asat \nsii$. Then, $M, \mint(w) \not\sat A \iimp B$, which implies that there exists a world $u \in W$ such that $\mint(w) \leq u$, $M, u \sat A$, and $M, u \not\sat B$. We now define a new interpretation $\mint'$ such that $M, \mint' \not\sat \ns \sbl \Gamma, A \sar B \sbr_{w}$. 


First, we set $\mint'(w) = u$. Second, we let $v_{1}, \ldots, v_{n}$ be (the names of) all children of $w$ (if they exist) with $v_{0}$ the (name of the) parent of $w$ (if it exists) in $\nsii$. Let $(w,\charx_{i},v_{i}) \in E$ for $i \in [n]$ and $(v_{0},\chary,w) \in E$ for $\tr(\nsii) = (T,E)$. Since $\mint(w) R_{\charx_{i}} \mint(v_{i})$ and $\mint(v_{0}) R_{\chary} \mint(w)$ hold in $M$ and $\mint(w) \leq u$, we know there exist $v_{i}'$ such that $\mint(v_{i}) \leq v_{i}'$ and $uR_{\charx_{i}}v_{i}'$ by condition (F1), and there exists a $v_{0}'$ such that $\mint(v_{0}) \leq v_{0}'$ and $v_{0}'R_{\chary}u$ by condition (F2) (see \dfn~\ref{def:bi-relational-model} for these conditions). Thus, if we set $\mint'(v_{i}) = v_{i}'$ and $\mint'(v_{0}) = v_{0}'$, then it follows that $\mint'(w) R_{\charx_{i}} \mint'(v_{i})$ and $\mint'(v_{0}) R_{\chary} \mint'(w)$. Moreover, since $M, \mint \sat \Gamma_{i}$ and $M, \mint \sat \Sigma$ for $(v_{i},\Gamma_{i} \sar \Delta_{i}), (v_{0},\Sigma \sar \Pi) \in T$, it follows that $M, \mint' \sat \Gamma_{i}$ and $M, \mint' \sat \Sigma$ by \lem~\ref{lem:persistence} as $\mint(v_{i}) \leq \mint'(v_{i})$ and $\mint(v_{0}) \leq \mint'(v_{0})$. We successively repeat this process until all components of $\nsii$ have been processed and $\mint'$ is fully defined over $\ns \sbl \Gamma, A \sar B \sbr_{w}$. One can confirm that $M, \mint' \not\sat \ns \sbl \Gamma, A \sar B \sbr_{w}$, showing the premise of $\iimpr$ $\axs$-invalid.

$\xdiar$. Suppose that $\nsii = \ns \sbl \Gamma \sar \xdia A, \Delta \sbr_{w} \sbl \Sigma \sar \Pi \sbr_{u}$ is $\axs$-invalid. By assumption, there exists an $\axs$-model $M$ and an $M$-interpretation $\mint$ such that $M, \mint \not\sat \nsii$, which further implies that $M, \mint(w) \not\sat \xdia A$. By the side condition on $\xdiar$, we know that $\prgr{\nsii} \models w \prpath{L_{\charx}} u$. By \lem~\ref{lem:paths-imply-relation}, $\mint(w)R\mint(u)$. Hence, $M, \mint(u) \not\sat A$, showing $M, \mint \not\sat \ns \sbl \Gamma \sar \Delta \sbr_{w} \sbl \Sigma \sar A, \Pi \sbr_{u}$.
\end{proof}

To prove completeness, one can show that all axioms in $\h\ikma$ are provable in $\calc$ and all rules of $\h\ikma$ are admissible (i.e., if the premises are provable, then the conclusion is provable). These proofs are a straightforward exercise, and so, we omit them. We note that \emph{modus ponens} ($\mathsf{mp}$) is admissible using the $\cut$ rule introduced in Section~\ref{sec:calc-properties}, which is itself an admissible rule (see \thm~\ref{thm:cut-elim}).

\begin{theorem}[Completeness]\label{thm:completeness-nested}  If $A$ is $\axs$-valid, then $\calc \proves A$.
\end{theorem}

\subsection{Separation and Related Systems.}\label{subsec:separation-related-logics} Our nested sequent systems also satisfy the \emph{separation property}, originally introduced in the context of nested calculi for tense logics~\cite{GorPosTiu11}. A calculus~$\calc$ has the separation property if and only if, whenever $\calc \proves A$ and the formula $A$ uses only modalities indexed with characters from $\albet' \subseteq \albet$, there exists a proof of $A$ in~$\calc$ that employs nested sequents using only characters drawn from $\albet'$, i.e., that do not mention any symbols in $\albet \setminus \albet'$. This property is significant in our setting because it entails that IGLs are conservative over IMLs. 

Let us restrict $\albet = \set{\chara}$ to a single character and define $\dia := \langle \chara \rangle$ and $\Box := [\chara]$. Moreover, we define $(\dia^{n}\Box A \iimp \Box^{k} A) \land (\dia^{k} A \iimp \Box^{n} \dia A), (\Box A \iimp \dia A) \in m(\axs)$ and $\Box A \iimp \dia A \in m(\axs)$ \iffi $(\langle \conv{\chara} \rangle^{n}\langle \chara \rangle^{k}A \iimp \langle \chara \rangle A) \land ([\chara] A \iimp [\conv{\chara}]^{n} [\chara]^{k}A), ([\chara] A \iimp \langle \chara \rangle A) \in \axs$. We refer to axioms of the latter form as \emph{Horn-Scott-Lemmon axioms (HSL axioms)} (cf.~\cite{Lyo21a}). In such a scenario, the calculus $\calc$ becomes a notational variant of the calculus $\mathsf{NIK}(m(\axs))$ presented in~\cite{Lyo21a} that is sound and complete for the extension of $\ik$ with the axioms $m(\axs)$. 

\begin{theorem}
$\ikma$ is a conservative extension of $\ik \cup m(\axs)$, i.e., if $\ent A$ and all modalities of $A$ are indexed by $\chara$, then $A$ is a theorem of $\ik \cup m(\axs)$.
\end{theorem}

\begin{proof} Suppose $\ent A$ and all modalities of $A$ are indexed by $\chara$. By completeness, $\calc \proves A$ and by the separation property we know that a proof $\prf$ exists in $\calc$ that only uses the single character $\chara$ in nested sequents. Hence, the proof is a notational variant of a proof in $\mathsf{NIK}(m(\axs))$ for the logic $\ik \cup m(\axs)$, showing that $A$ is a theorem of $\ik \cup m(\axs)$.
\end{proof}

Moreover, if we fix $\albet := \set{\chara, \conv{\chara}}$, and define $\mathsf{F} := \langle \chara \rangle$, $\mathsf{P} := \langle \conv{\chara} \rangle$, $\mathsf{G} := [\chara]$, and $\mathsf{H} := [\conv{\chara}]$, then $\calc$ becomes a sound and complete calculus for Ewald's tense logic $\mathsf{IKt}$~\cite{Ewa86} extended with the axioms in $\axs$. We refer to such a logic as an \emph{intuitionistic tense logic (ITL)} and refer to a set $\axs$ of axioms within the signature $\albet := \set{\chara, \conv{\chara}}$ as \emph{intuitionistic tense path axioms (ITP axioms)}.

As a final observation, we also remark that our systems are single-conclusioned (i.e., intuitionistic) variants of the nested sequent systems for classical grammar logics, introduced by Tiu et al.~\cite{TiuIanGor12}. If we relax the single-conclusioned restriction imposed on nested sequents in $\calc$, this gives rise to a nested sequent calculus that is sound and complete for the classical grammar logic $\km$ extended with the axioms in $\axs$. In the classical setting, each IPA $(\langle \stra \rangle A \iimp \langle \charx \rangle A) \land ([ \charx ] A \iimp [\stra] A)$ is equivalent to the \emph{path axiom} $\langle \stra \rangle A \iimp \langle \charx \rangle A$ (cf.~\cite{TiuIanGor12}). 

\section{Proof-Theoretic Properties}\label{sec:calc-properties}

In this section, we show that each nested calculus $\calc$ satisfies a wide array of fundamental properties, which we leverage in our proof of syntactic cut-admissibility. In particular, we show that various rules are \emph{(height-preserving) admissible}, meaning if the premises of the rule have proofs (of height $h_{1}, \ldots, h_{n}$), then the conclusion of the rule has a proof (of height $h \leq \max\{h_{1}, \ldots, h_{n}\}$). With the exception of $\iimpl$, we also prove all left rules of $\calc$ \emph{height-preserving invertible}. If we let $\ru^{-1}_{i}$ be the $i$-inverse of the rule $\ru$ whose conclusion is the $i^{th}$ premise of the $n$-ary rule $\ru$ and premise is the conclusion of $\ru$, then we say that $\ru$ is \emph{(height-preserving) invertible} \iffi $\ru^{-1}_{i}$ is (height-preserving) admissible for each $i \in [n]$. (NB. We refer to a height-preserving admissible and height-preserving invertible rule as \emph{hp-admissible} and \emph{hp-invertible}, respectively.)

\begin{figure}[t]
\centering

\begin{tikzpicture}[>=stealth']
  \node (d) at (0, 0.9)  {$\ew$};
  \node (a) at (0,-0.9)  {$\wkl$};

  \node (b) at (3, 0)    {Lem.~\ref{lem:hp-inver}};

  \node (c) at (6, 0.9)  {$\ctrl$};
  \node (g) at (6,-0.9)  {$\ec$};

  \draw[->] (d) -- (b);
  \draw[->] (a) -- (b);

  \draw[->] (b) -- (c);
  \draw[->] (b) -- (g);

  \draw[->] (g) to[bend right=25] (c);
\end{tikzpicture}


\caption{A diagram depicting which hp-admissibility (and hp-invertibility) results are sufficient to prove others. An arrow from one rule name to another indicates that the hp-admissibility of the source rule is sufficient to prove the hp-admissibility of the target rule. Note that `Lem.~\ref{lem:hp-inver}' indicates the hp-invertibility of certain rules in $\calc$.\label{fig:admiss-invert-dependencies}}


\end{figure}

We will show all rules in \fig~\ref{fig:structural-rules} hp-admissible in $\calc$, i.e., we prove the \emph{right bottom rule} $\botr$, \emph{necessitation rule} $\nec$, \emph{left weakening rule} $\wkl$, \emph{right weakening rule} $\wkr$, \emph{external weakening rule} $\ew$, \emph{left contraction rule} $\ctrl$, \emph{external contraction rule} $\ec$, and \emph{shift rule} $\shift$ hp-admissible. These hp-admissibility results are crucial for proving cut-admissibility. As we will see below, various hp-admissibility results will be used to prove further structural rules hp-admissible. To make the dependencies between these various lemmas clear, we have provided a diagram in \fig~\ref{fig:admiss-invert-dependencies} that displays which hp-admissible rules are sufficient to prove other rules hp-admissible; rules omitted from the diagram are hp-admissible of their own accord.

The shift rule $\shift$ is a novel contribution of this paper and is interesting as it unifies reasoning with IPAs in a single rule. The $\shift$ rule utilizes the $\odot$ operation, which is defined as follows: Let $\ns := \Gamma \sar \Delta, (\charx_{1})[\nsii_{1}], \ldots, (\charx_{n})[\nsii_{n}]$ and $\nsiii := \Sigma \sar \Pi, (\charx_{n{+}1})[\nsiv_{n{+}1}], \ldots, (\charx_{n{+}k})[\nsiv_{n{+}k}]$ be nested sequents, then 
$$
\ns \odot \nsiii := \Gamma, \Sigma \sar \Delta, \Pi, (\charx_{1})[\nsii_{1}], \ldots, (\charx_{n})[\nsii_{n}], (\charx_{n{+}1})[\nsiv_{n{+}1}], \ldots, (\charx_{n{+}k})[\nsiv_{n{+}k}].
$$
This rule streamlines our proofs because it does not require one to use a special structural rule for \emph{each} IPA, which contrasts with prior approaches~\cite{Bru09,LyoOrl23,Str13}. Moreover, as pointed out by Brünnler~\cite{Bru09}, having a distinct structural rule for each axiom extending the base logic can lead to a loss of modularity, requiring the addition of even more structural rules to secure completeness. This allows us to prove the hp-admissibility of a single, generic rule that simplifies our proof of cut-admissibility, making it uniform over the class of logics we consider. 

In the following subsection, we prove a sequence of hp-admissibility and hp-invertibility results, which will be sufficient to establish cut-admissibility in the \sect~\ref{subsec:cut-elim}.

\begin{figure}[t]
\noindent

\begin{center}
\begin{tabular}{c c c}
\AxiomC{$\ns \sbl \Gamma \sar \bot \sbr$}
\RightLabel{$\botr$}
\UnaryInfC{$\ns \sbl \Gamma \sar \emptyset \sbr$}
\DisplayProof

&

\AxiomC{$\ns \sbl \Gamma \sar \Delta \sbr$}
\RightLabel{$\wkl$}
\UnaryInfC{$\ns \sbl \Gamma, \Sigma \sar \Delta \sbr$}
\DisplayProof

&

\AxiomC{$\ns^{\downarrow} \sbl \Gamma \sar \emptyset \sbr$}
\RightLabel{$\wkr$}
\UnaryInfC{$\ns \sbl \Gamma\sar A \sbr$}
\DisplayProof
\end{tabular}
\end{center}


\begin{center}
\begin{tabular}{c c c}
\AxiomC{$\ns$}
\RightLabel{$\nec$}
\UnaryInfC{$\emptyset \sar \emptyset, (\charx) \bl \ns \br$}
\DisplayProof

&

\AxiomC{$\ns \sbl \Gamma, A, A \sar \Delta \sbr$}
\RightLabel{$\ctrl$}
\UnaryInfC{$\ns \sbl \Gamma, A \sar \Delta \sbr$}
\DisplayProof

&

\AxiomC{$\ns \sbl \Gamma \sar \Delta \sbr$}
\RightLabel{$\ew$}
\UnaryInfC{$\ns \sbl \Gamma \sar \Delta, (\charx)[\emptyset \sar \emptyset] \sbr$}
\DisplayProof
\end{tabular}
\end{center}


\begin{center}
\begin{tabular}{c c}
\AxiomC{$\ns \sbl \Gamma \sar \Delta, (\charx)[\nsii], (\charx)[\nsii] \sbr$}
\RightLabel{$\ec$}
\UnaryInfC{$\ns \sbl \Gamma \sar \Delta, (\charx)[\nsii] \sbr$}
\DisplayProof

&

\AxiomC{$\ns \sbl \Gamma \sar \Delta, (\charx)[\nsii] \sbr_{w} \sbl  \nsiii \sbr_{v}$}
\RightLabel{$\shift^{\dag(\axs)}$}
\UnaryInfC{$\ns \sbl \Gamma \sar \Delta \sbr_{w} \sbl  \nsiii \odot \nsii \sbr_{v}$}
\DisplayProof
\end{tabular}
\end{center}


\begin{flushleft}
\textbf{Side Conditions:}\\
$\dag(\axs):= w \prpath{L_{\charx}} u$ with $L_{\charx} = \glang(\charx)$.
\end{flushleft}

\caption{Height-preserving admissible rules in $\calc$.\label{fig:structural-rules}}
\end{figure}

\subsection{Admissibility and Invertibility}

The lemma below establishes that a generalized form of the $\id$ rule is admissible in $\calc$. This property ensures that provable nested sequents are closed under substitutions, which is essential for completeness (cf.~\cite[Lemma 12]{Lyo21thesis}). It is straightforward to prove the result by induction on the length of $A$.

\begin{lemma}\label{lem:gen-id}
For $A \in \lang$, $\ns\set{\Gamma, A \sar A}$ is provable in $\calc$.
\end{lemma}

The remaining hp-admissibility results are all proven by induction on the height of the given proof $\prf$.

\begin{lemma}
The $\botr$, $\nec$, $\wkl$, $\wkr$, and $\ew$ rules are hp-admissible in $\calc$.
\end{lemma}

\begin{proof} Let $\ru \in \set{\botr, \nec, \wkl, \wkr, \ew}$. We make a case distinction on the last rule $\ru' \in \calc$ applied in the given proof $\prf$. Observe that applying $\ru$ (if possible) to an $\id$ or $\botl$ instance gives another instance of that rule, which resolves the base case. For the inductive step, each case is resolved by applying IH (i.e., the induction hypothesis) to the premise(s) of $\ru'$ followed by $\ru'$.
\end{proof}

As depicted in \fig~\ref{fig:admiss-invert-dependencies}, the hp-admissibility of the two contraction rules $\ctrl$ and $\ec$ depends on the hp-invertibility of 
certain rules in $\calc$. In contrast to the nested systems for classical grammar logics~\cite{TiuIanGor12}, not all rules in $\calc$ are invertible; e.g., the $\iimpl$ rule is \emph{not} invertible in the left premise as witnessed by the following example:
\begin{center}
\AxiomC{$p \iimp q \sar p$}
\AxiomC{$q \sar p \iimp q$}
\RightLabel{$\iimpl$}
\BinaryInfC{$p \iimp q \sar p \iimp q$}
\DisplayProof
\end{center}
One can still establish that all remaining left rules in $\calc$ are hp-invertible, which is sufficient to obtain the hp-admissibility of the contraction rules.

\begin{lemma}\label{lem:hp-inver}
The $\disl$, $\conl$, $\dial$, $\boxl$, and $\ddr$ rules are hp-invertible in $\calc$.
\end{lemma}

\begin{proof} The hp-invertibility of $\xboxl$ and $\ddr$ follow from the hp-admissibility of $\wkl$ and $\ew$. The remaining cases are argued as usual by induction on the height of the given proof.
\end{proof}

\begin{lemma}
The $\ctrl$ and $\ec$ rules are hp-admissible in $\calc$.
\end{lemma}

\begin{proof} We prove the result by a simultaneous induction on the height of the given proof. The base cases are simple since any application of $\ctrl$ or $\ec$ to an initial rule gives another instance of that rule; hence, we focus on the inductive step. For the inductive step, we show the $\xdial$ case, noting that the other cases are simple or similar.

We consider the case where $\ctrl$ is preceded by an application of $\xdial$ and where the principal formula of $\xdial$ is auxiliary in $\ctrl$. This case is shown below left and is resolved as shown below right. We first use the hp-invertibility of $\xdial$ (\lem~\ref{lem:hp-inver}), then apply IH relative to $\ec$, and last apply the $\xdial$ rule.
\begin{center}
\begin{tabular}{c c} 
\AxiomC{$\dotprf{\prf}$}
\UnaryInfC{$\ns \sbl \Gamma, \xdia A \sar \Delta, (\charx)[ A \sar \emptyset ] \sbr$}
\RightLabel{$\xdial$}
\UnaryInfC{$\ns \sbl \Gamma, \xdia A, \xdia A \sar \Delta \sbr$}
\RightLabel{$\ctrl$}
\UnaryInfC{$\ns \sbl \Gamma, \xdia A \sar \Delta \sbr$}
\DisplayProof

&

\AxiomC{$\dotprf{\prf'}$}
\UnaryInfC{$\ns \sbl \Gamma, \xdia A \sar \Delta, (\charx)[ A \sar \emptyset ] \sbr$}
\RightLabel{$\xdial_{1}^{-1}$}
\UnaryInfC{$\ns \sbl \Gamma \sar \Delta, (\charx)[ A \sar \emptyset ], (\charx)[ A \sar \emptyset ] \sbr$}
\RightLabel{$\ec$}
\UnaryInfC{$\ns \sbl \Gamma \sar \Delta, (\charx)[ A \sar \emptyset ] \sbr$}
\RightLabel{$\xdial$}
\UnaryInfC{$\ns \sbl \Gamma, \xdia A \sar \Delta \sbr$}
\DisplayProof
\end{tabular}
\end{center}
This case shows that the hp-admissibility of $\ctrl$ depends on that of $\ec$. 
\end{proof}

\begin{lemma}
The $\shift$ rule is hp-admissible in $\calc$.
\end{lemma}

\begin{proof} If $\shift$ is applied to an initial rule, then the result is another instance of an initial rule, showing the base case of induction. For the inductive step, we consider the cases where $\shift$ is applied to the conclusion of an $\xdiar$ rule and $\xboxr$ rule. The remaining cases are analogous or straightforward.

$\xdiar$. Let us consider an application of $\xdiar$ followed by an application of $\shift$, as shown below. 
\begin{center}
\AxiomC{$\dotprf{\prf}$}
\UnaryInfC{$\ns_{0} \sbl \Gamma \sar \emptyset \sbr_{w} \sbl  \Sigma \sar A \sbr_{u} \sbl \Phi \sar \emptyset, (\chary)[\nsii] \sbr_{v} \sbl \nsiii \sbr_{z}$}
\RightLabel{$\xdiar$}
\UnaryInfC{$\ns_{0} \sbl \Gamma \sar \xdia A \sbr_{w} \sbl  \Sigma \sar \emptyset \sbr_{u} \sbl \Phi \sar \emptyset, (\chary)[\nsii] \sbr_{v} \sbl \nsiii \sbr_{z}$}
\RightLabel{$\shift$}
\UnaryInfC{$\ns_{1} \sbl \Gamma \sar \xdia A \sbr_{w} \sbl  \Sigma \sar \emptyset \sbr_{u} \sbl \Phi \sar \emptyset \sbr_{v} \sbl \nsiii \odot \nsii \sbr_{z}$}
\DisplayProof
\end{center}
Let $L_{\charx} = \glang(\charx)$ and $L_{\chary} = \glang(\chary)$. The side condition on the $\xdiar$ rule ensures that $w \prpath{L_{\charx}} u$ in $\ns_0$ and the side condition on rule $\shift$ ensures that $v \prpath{L_{\chary}} z$ in $\ns_0$. We have to show that both side conditions continue to hold after we permute $\shift$ above $\xdiar$, thus giving the proof shown below.
\begin{center}
\AxiomC{$\dotprf{\prf}$}
\UnaryInfC{$\ns_{0} \sbl \Gamma \sar \emptyset \sbr_{w} \sbl  \Sigma \sar A \sbr_{u} \sbl \Phi \sar \emptyset, (\chary)[\nsii] \sbr_{v} \sbl \nsiii \sbr_{z}$}
\RightLabel{$\shift$}
\UnaryInfC{$\ns_{0} \sbl \Gamma \sar \emptyset \sbr_{w} \sbl  \Sigma \sar A \sbr_{u} \sbl \Phi \sar \emptyset \sbr_{v} \sbl \nsiii \odot \nsii \sbr_{z}$}
\RightLabel{$\xdiar$}
\UnaryInfC{$\ns_{1} \sbl \Gamma \sar \xdia A \sbr_{w} \sbl  \Sigma \sar \emptyset \sbr_{u} \sbl \Phi \sar \emptyset \sbr_{v} \sbl \nsiii \odot \nsii \sbr_{z}$}
\DisplayProof
\end{center}
If the $w$-component and the $u$-component both occur in $\nsii$ or neither occurs in $\nsii$, then $\shift$ can be permuted above $\xdiar$ because the two rules do not interact. Let us consider the case where the $u$-component occurs in $\nsii$ and the $w$-component does not. (NB. The case where the $w$-component occurs in $\nsii$ and the $u$-component does not is proven similarly.) Let $v'$ be the name of the component that serves as the root of $\nsii$. It follows that the propagation path from $w$ to $u$ in $\ns_{0}$ is of the form $w \prpath{\stra} v \prpath{\chary} v' \prpath{\strb} u$ with $\stra \chary \strb \in L_{\charx}$. By the side condition of the $\shift$ rule, we know there exists a string $\strc \in L_{\chary}$ such that $v \prpath{\strc} z$. Since $L_{\chary}  = \glang(\chary)$, this implies that $\chary \longrightarrow^{*}_{\g{\axs}} \strc$, which further implies that $\stra \strc \strb \in L_{\charx}$. This string corresponds to the propagation path $w \prpath{\stra} v \prpath{\strc} z \prpath{\strb} u$ in $\ns_{1}$, showing that after $\shift$ has been applied the side condition on $\xdiar$ still holds, and thus, the two rules may indeed be permuted as shown above.

$\xboxr$. Suppose we have an instance of $\xboxr$ followed by an instance of $\shift$, as shown below. 
\begin{center}
\AxiomC{$\dotprf{\prf}$}
\UnaryInfC{$\ns \sbl \Gamma \sar (\charx)[\emptyset \sar A] \sbr_{w} \sbl  \Sigma \sar \emptyset, (\chary)[\nsii] \sbr_{u} \sbl \nsiii \sbr_{v}$}
\RightLabel{$\xboxr$}
\UnaryInfC{$\ns \sbl \Gamma \sar \xbox A \sbr_{w}  \sbl  \Sigma \sar \emptyset, (\chary)[\nsii] \sbr_{u} \sbl \nsiii \sbr_{v}$}
\RightLabel{$\shift$}
\UnaryInfC{$\ns \sbl \Gamma \sar \xbox A \sbr_{w} \sbl  \Sigma \sar \emptyset \sbr_{u} \sbl \nsiii \odot \nsii \sbr_{v}$}
\DisplayProof
\end{center}
Observe that we can simply permute the $\shift$ application above the $\boxr$ application to obtain the desired result.
\end{proof}

\subsection{Syntactic Cut-Admissibility}\label{subsec:cut-elim}

As discussed in prior sections, Tiu et al.~\cite{TiuIanGor12} provided (both shallow and deep) nested sequent systems for classical grammar logics. A cut admissibility algorithm was given for the shallow systems by utilizing a strategy akin to Belnap's general cut admissibility theorem for display calculi~\cite[Theorem 4.4]{Bel82}. Cut admissibility was then shown for the deep nested systems in a roundabout fashion by proving a syntactic correspondence between the shallow and deep nested systems. It is conceivable that we could achieve cut admissibility in the intuitionistic setting by employing a similar (indirect) strategy; however, we opt for a syntactic cut admissibility algorithm that eliminates cuts \emph{directly within} our nested systems. We remark that this approach can be adapted to the classical setting--giving a cut admissibility procedure that operates directly within the deep nested sequent proofs of Tiu et al.~\cite{TiuIanGor12}.

In the deep-inference setting (and in contrast to the shallow setting), the proof of cut admissibility is more complex because we must keep track of propagation paths and ensure the satisfaction of side conditions on propagation rules after upward permutations of cuts (see \thm~\ref{thm:cut-elim} below). Still, beyond being direct, our cut admissibility proof also has the advantage that it is \emph{uniform} over the class of logics we consider. 
This is primarily due to our use of the hp-admissible shift rule $\shift$, which unifies reasoning with IPAs in a single rule. This also allows for us to circumvent the use of \emph{ad hoc} structural rules to deal with special cut admissibility cases, as is the case in other settings~\cite{Bru09,LyoOrl23,Str13}.

\begin{theorem}[Cut-Admissibility]\label{thm:cut-elim}
The $\cut$ rule is admissible in $\calc$.
\begin{center}
\AxiomC{$\ns^{\da} \sbl \Gamma \sar A \sbr$}
\AxiomC{$\ns \sbl \Gamma, A \sar \Delta \sbr$}
\RightLabel{$\cut$}
\BinaryInfC{$\ns \sbl \Gamma \sar \Delta \sbr$}
\DisplayProof
\end{center}
\end{theorem}

\begin{proof} We prove the result by simultaneous induction on the lexicographic ordering of pairs of the form $(\len{A},h_{1}+h_{2})$, where $\len{A}$ is the length of the cut formula $A$, and $h_{1}$ and $h_2$ are the heights of the proofs of the left and right premises of $\cut$, respectively. We assume w.l.o.g. that $\cut$ occurs only once in the given proof as last inference; the general result follows by successively applying the cut admissibility procedure to topmost instances of $\cut$ in a given proof. Note that $\cut$ can always be permuted above applications of $\ddr$ by utilizing the hp-admissibility of $\ew$, and thus, we may omit consideration of $\ddr$ below as all such cases hold.

We organize the proof into three exhaustive cases: (1) at least one premise of $\cut$ is an instance of $\id$ or $\botl$, (2) the cut formula is not principal in at least one premise, and (3) the cut formula is principal in both premises.

(1) 
Let us consider the case where the left premise of $\cut$ is an instance of $\botl$ or $\id$. 
If the left premise of $\cut$ is an instance of $\botl$, then the conclusion of $\cut$ will be an instance of $\botl$, showing that $\cut$ can be replaced by an application of the $\botl$ rule. 

Let us now suppose that the left premise of $\cut$ is an instance of $\id$. Then, 
the principal (output) formula $p$ of $\id$ is the cut formula. The $\cut$ is of the form shown below left and can be removed by applying the hp-admissible $\ctrl$ rule shown below right.
\begin{center}
\begin{tabular}{c c}
\AxiomC{ } 
\RightLabel{$\id$}
\UnaryInfC{$\ns^{\da}\sbl \Gamma, p \sar p \sbr$}

\AxiomC{$\dotprf{\prf}$}
 \UnaryInfC{$\ns\sbl \Gamma, p, p \sar \Delta \sbr$}
 \RightLabel{$\cut$}
\BinaryInfC{$\ns\sbl \Gamma, p \sar \Delta \sbr$}
\DisplayProof

&

\AxiomC{$\dotprf{\prf}$}
 \UnaryInfC{$\ns\sbl \Gamma, p, p \sar \Delta \sbr$}
 \RightLabel{$\ctrl$}
\UnaryInfC{$\ns\sbl \Gamma, p \sar \Delta \sbr$}
\DisplayProof
\end{tabular}
\end{center}

Next, we suppose that the right premise of $\cut$ is an instance of $\botl$ or $\id$. In the former case, if the principal formula $\bot$ is not the cut formula, then the conclusion will be an instance of $\botl$ and may therefore be proven with an instance of $\botl$. If the principal formula $\bot$ is the cut formula, as shown below left, then the case may be resolved with the hp-admissible $\botr$ and $\wkr$ rules, as shown below right.
\begin{center}
\begin{tabular}{c c}
\AxiomC{$\dotprf{\prf}$}
 \UnaryInfC{$\ns^{\da}\sbl \Gamma \sar \bot \sbr$}

\AxiomC{ } 
\RightLabel{$\botl$}
\UnaryInfC{$\ns\sbl \Gamma, \bot \sar \Delta \sbr$}

 \RightLabel{$\cut$}
\BinaryInfC{$\ns\sbl \Gamma \sar \Delta \sbr$}
\DisplayProof

&

\AxiomC{$\dotprf{\prf}$}
 \UnaryInfC{$\ns\sbl \Gamma \sar \bot \sbr$}
 \RightLabel{$\botr$}
  \UnaryInfC{$\ns\sbl \Gamma \sar \emptyset \sbr$}
 \RightLabel{$\wkr$}
\UnaryInfC{$\ns\sbl \Gamma \sar \Delta \sbr$}
\DisplayProof
\end{tabular}
\end{center}
If the right premise of $\cut$ is an instance of $\id$, then there are two cases to consider: first, if the principal formula $p$ in the right premise is not the cut formula, then the conclusion of $\cut$ will be an instance of $\id$, so the cut can be replaced by an $\id$ application. Second, if the principal formula $p$ in the right premise is the cut formula, then the conclusion of $\cut$ will be identical to its left premise, so the $\cut$ can be replaced by the proof of its left premise.

(2) Let us suppose that the cut formula is not principal in the left premise of $\cut$; the case where the cut formula is not principal in the right premise of $\cut$ is argued similarly. First, let us consider the case where $\ru$ is a left rule other than $\iimpl$ that proves the left premise of $\cut$, as shown below left. We also suppose that $\ru$ is a unary rule as the binary case is similar. To resolve the case, we apply the hp-invertibility of $\ru$ to the conclusion of $\prf_{2}$, then cut with the premise of $\ru$, and then apply $\ru$ after $\cut$, as shown below right.
\begin{center}
\begin{tabular}{c c}
\AxiomC{$\dotprf{\prf_{1}}$}
\UnaryInfC{$\nsii^{\da}\sbl \Gamma' \sar A \sbr$}
\RightLabel{$\ru$}
\UnaryInfC{$\ns^{\da}\sbl \Gamma \sar A \sbr$}

\AxiomC{$\dotprf{\prf_{2}}$}
\UnaryInfC{$\ns\sbl \Gamma, A \sar \Delta \sbr$}
\RightLabel{$\cut$}
\BinaryInfC{$\ns\sbl \Gamma \sar \Delta \sbr$}
\DisplayProof

&

\AxiomC{$\dotprf{\prf_{1}}$}
\UnaryInfC{$\nsii^{\da}\sbl \Gamma' \sar A \sbr$}

\AxiomC{$\dotprf{\prf_{2}}$}
\UnaryInfC{$\ns\sbl \Gamma, A \sar \Delta \sbr$}
\RightLabel{$\ru_{1}^{-1}$} 
\UnaryInfC{$\nsii\sbl \Gamma', A \sar \Delta \sbr$}
\RightLabel{$\cut$}
\BinaryInfC{$\nsii\sbl \Gamma' \sar \Delta \sbr$}
\RightLabel{$\ru$}
\UnaryInfC{$\ns\sbl \Gamma \sar \Delta \sbr$}
\DisplayProof
\end{tabular}
\end{center}

Let us suppose now that the left premise of $\cut$ is proven using $\iimpl$ and suppose w.l.o.g. that the rule affects the same component as the cut formula. As the cut formula $A$ is assumed not principal in the left premise, our $\cut$ is of the form shown below top. It can be resolved as shown below bottom by shifting $\cut$ up the right premise of $\iimpl$.
\begin{center}
\AxiomC{$\dotprf{\prf_{1}}$}
\UnaryInfC{$\ns^{\da}\sbl \Gamma, B  \iimp C \sar B \sbr$}
\AxiomC{$\dotprf{\prf_{2}}$}
\UnaryInfC{$\ns\sbl \Gamma, C \sar A \sbr$}
\RightLabel{$\iimpl$}
\BinaryInfC{$\ns^{\da}\sbl \Gamma, B \imp C \sar A \sbr$}

\AxiomC{$\dotprf{\prf_{3}}$}
\UnaryInfC{$\ns\sbl \Gamma, A \sar \Delta \sbr$}
\RightLabel{$\cut$}
\BinaryInfC{$\ns\sbl \Gamma, B \imp C \sar \Delta \sbr$}
\DisplayProof
\end{center}
\begin{center}
\AxiomC{$\dotprf{\prf_{1}}$}
\UnaryInfC{$\ns^{\da}\sbl \Gamma, B  \iimp C \sar B \sbr$}

\AxiomC{$\dotprf{\prf_{2}}$}
\UnaryInfC{$\ns^{\da}\sbl \Gamma, C \sar A \sbr$}
\AxiomC{$\dotprf{\prf_{3}}$}
\UnaryInfC{$\ns\sbl \Gamma, A \sar \Delta \sbr$}
\RightLabel{$\wkl$}
\UnaryInfC{$\ns\sbl \Gamma, C, A \sar \Delta \sbr$}
\RightLabel{$\cut$}

\BinaryInfC{$\ns\sbl \Gamma, C \sar \Delta \sbr$}
\RightLabel{$\iimpl$}

\BinaryInfC{$\ns\sbl \Gamma, B \imp C \sar \Delta \sbr$}
\DisplayProof
\end{center}

(3) We consider two non-trivial cases and note that the remaining cases are simple or similar. First, we consider the case where the cut formula is principal in an application of $\iimpr$ in the left premise and an application of $\iimpl$ in the right premise. 
\begin{center}
\AxiomC{$\dotprf{\prf_{1}}$}
\UnaryInfC{$\ns\sbl \Gamma, A \sar B \sbr$}
\RightLabel{$\iimpr$}
\UnaryInfC{$\ns\sbl \Gamma \sar A \imp B \sbr$}

\AxiomC{$\dotprf{\prf_{2}}$}
\UnaryInfC{$\ns^{\da} \sbl \Gamma, A \iimp B \sar A \sbr$}

\AxiomC{$\dotprf{\prf_{3}}$}
\UnaryInfC{$\ns \sbl \Gamma, B \sar \Delta \sbr$}
\RightLabel{$\iimpl$}
\BinaryInfC{$\ns \sbl \Gamma, A \iimp B \sar \Delta \sbr$}
\RightLabel{$\cut$}
\BinaryInfC{$\ns \sbl \Gamma \sar \Delta \sbr$}
\DisplayProof
\end{center}
To remove the $\cut$, we first reduce the height of the $\cut$ by applying a $\cut$ between the conclusion of $\iimpr$ and the conclusion of $\prf_{2}$. This gives a cut-free proof of $\ns \sbl \Gamma \sar A \sbr$ by IH since $h_{1} + h_{2}$ has decreased. We then apply a $\cut$ between the conclusion of $\prf_{1}$ and the former nested sequent to obtain a cut-free proof of $\ns\sbl \Gamma \sar B \sbr$, which can then be cut with $\ns \sbl \Gamma, B \sar \Delta \sbr$. These last two cuts can be eliminated since the cut formulae are of a strictly smaller length.


Last, we show how to resolve the case where the cut formula is principal in an application of $\xboxr$ in the left premise and an application of $\xboxl$ in the right premise. 
\begin{center}
\AxiomC{$\dotprf{\prf_{1}}$}
\UnaryInfC{$\ns \sbl \Gamma \sar (\charx) [ \emptyset \sar A ] \sbr_{w} \sbl \Sigma \sar \emptyset \sbr_{u}$}
\RightLabel{$\xboxr$}
\UnaryInfC{$\ns \sbl \Gamma \sar \xbox A \sbr_{w} \sbl \Sigma \sar \emptyset \sbr_{u}$}

\AxiomC{$\dotprf{\prf_{2}}$}
\UnaryInfC{$\ns \sbl \Gamma, \xbox A \sar \Delta \sbr_{w} \sbl \Sigma, A \sar \Pi \sbr_{u}$}
\RightLabel{$\xboxl$}
\UnaryInfC{$\ns \sbl \Gamma, \xbox A \sar \Delta \sbr_{w} \sbl \Sigma \sar \Pi \sbr_{u}$}

\RightLabel{$\cut$}
\BinaryInfC{$\ns \sbl \Gamma \sar \Delta \sbr_{w} \sbl \Sigma \sar \Pi \sbr_{u}$}
\DisplayProof
\end{center}
We first reduce the height of the $\cut$ as shown below, relying on the hp-admissible $\wkl$ rule. We let $\prf_{3}$ denote the resulting proof.
\begin{center}
\AxiomC{$\dotprf{\prf_{1}}$}
\UnaryInfC{$\ns \sbl \Gamma \sar (\charx) [ \emptyset \sar A ] \sbr_{w} \sbl \Sigma \sar \emptyset \sbr_{u}$}
\RightLabel{$\xboxr$}
\UnaryInfC{$\ns \sbl \Gamma \sar \xbox A \sbr_{w} \sbl \Sigma \sar \emptyset \sbr_{u}$}
\RightLabel{$\wkl$}
\UnaryInfC{$\ns \sbl \Gamma \sar \xbox A \sbr_{w} \sbl \Sigma, A \sar \emptyset \sbr_{u}$}

\AxiomC{$\dotprf{\prf_{2}}$}
\UnaryInfC{$\ns \sbl \Gamma, \xbox A \sar \Delta \sbr_{w} \sbl \Sigma, A \sar \Pi \sbr_{u}$}

\RightLabel{$\cut$}
\BinaryInfC{$\ns \sbl \Gamma \sar \Delta \sbr_{w} \sbl \Sigma, A \sar \Pi \sbr_{u}$}
\DisplayProof
\end{center}
Due to the side condition on the $\xboxl$ rule, we know that $w \prpath{L_{\charx}} u$. Therefore, we may apply the hp-admissible $\shift$ rule as shown below, and then apply a $\cut$ on the formula $A$, which is of a smaller length.
\begin{center}
\AxiomC{$\dotprf{\prf_{1}}$}
\UnaryInfC{$\ns^{\da} \sbl \Gamma \sar (\charx) [ \emptyset \sar A ] \sbr_{w} \sbl \Sigma \sar \emptyset \sbr_{u}$}
\RightLabel{$\shift$}
\UnaryInfC{$\ns^{\da} \sbl \Gamma \sar \emptyset \sbr_{w} \sbl \Sigma \sar A \sbr_{u}$}

\AxiomC{$\dotprf{\prf_{3}}$}

\RightLabel{$\cut$}
\BinaryInfC{$\ns \sbl \Gamma \sar \Delta \sbr_{w} \sbl \Sigma \sar \Pi \sbr_{u}$}
\DisplayProof
\end{center}
This concludes the proof of cut-admissibility.
\end{proof}

\section{Undecidability via Negative Translations}\label{sec:interpolation}

In this section, we prove that the general validity problem for IGLs is undecidable. This is achieved by adapting the \emph{negative translation} of Gödel, Gentzen, and Kolmogorov (cf.~\cite{Bus98}) to our setting, yielding a faithful embedding of CGLs into IGLs. As the general validity problem is undecidable for CGLs (see Baldoni et al.~\cite[Corollary 2]{BalGioMar98}), this reduction implies undecidability in the intuitionistic setting as well.

Recall that at the end of Section~\ref{sec:nested-calculi}, we made the observation that dropping the single-conclusion restriction on nested sequents transforms every nested calculus for an IGL into one for a CGL. To make this precise, we define a \emph{classical nested sequent (CNS)} inductively as follows:
\begin{enumerate}

\item[$(1)$] $\Gamma \sar \Delta$ is a CNS, where $\Gamma$ and $\Delta$
are finite multisets of formulae from $\lang$;

\item[$(2)$] If $\ns_{1}, \ldots, \ns_{n}$ are CNSs, then
$\Gamma \sar \Delta, (\charx_{1})[\ns_{1}], \ldots,
(\charx_{n})[\ns_{n}]$ is a CNS.

\end{enumerate}
We define the calculus $\cglcalc$ to be the variant of $\calc$ whose rule applications are permitted to operate on CNSs rather than single-conclusioned nested sequents. It is straightforward to verify that each classical calculus $\cglcalc$ is a notational variant of a nested sequent calculus of Tiu et al.~\cite{TiuIanGor12}. Moreover, in the context of $\cglcalc$, we assume that all structural rules in Figure~\ref{fig:structural-rules} and $\cut$ have their single-conclusion restriction lifted and note that all rules are admissible by the work of Tiu et al.~\cite{TiuIanGor12}. For more information on the proof theory and semantics of classical (context-free) grammar logics, see~\cite{BalGioMar98,CerPen88,TiuIanGor12}.

\begin{definition}[Negative Translation] We define the \emph{negative translation} $\ddn{(\cdot)} : \lang \to \lang$ of formulae by recursion on the length of a formula. For any propositional atom $p \in \prop$, $\ddn{p} := \neg \neg p$. The remaining cases are listed below:
\begin{multicols}{2}
\begin{itemize}


\item $\ddn{\bot} := \neg \neg \bot$

\item $\ddn{(A \lor B)} := \neg(\neg\ddn{A} \land \neg\ddn{B})$

\item $\ddn{(A \land B)} := \ddn{A} \land \ddn{B}$

\item $\ddn{(A \iimp B)} := \ddn{A} \iimp \ddn{B}$

\item $\ddn{(\xdia A)} := \neg \xbox \neg \ddn{A}$

\item $\ddn{(\xbox A)} := \xbox \ddn{A}$

\end{itemize}
\end{multicols}
\end{definition}


\begin{lemma}\label{lem:double-neg-implies-ddn-trans}
For any $A \in \lang$, $\neg \neg \ddn{A} \sar \ddn{A}$ is provable in $\calc$.
\end{lemma}

\begin{proof} By induction on the length of $A$. We show the case where $A$ is of the form $\xbox B$. In the proof below, we invoke the admissibility of $\cut$ in $\calc$ (Theorem~\ref{thm:cut-elim}). The top left leaf sequent is clearly provable in $\calc$, we let $\ns := \bot \sar \ (\charx) [ \neg \ddn{B} \sar \bot ]$, and the top right leaf sequent is provable by IH. 
\begin{center}
\AxiomC{$\neg \neg \xbox \ddn{B}, \xbox \ddn{B} \sar \bot, (\charx) [ \neg \ddn{B}, \ddn{B} \sar \ ]$}
\RightLabel{$\xboxl$}
\UnaryInfC{$\neg \neg \xbox \ddn{B}, \xbox \ddn{B} \sar \bot, (\charx) [ \neg \ddn{B} \sar \ ]$}
\RightLabel{$\iimpr$}
\UnaryInfC{$\neg \neg \xbox \ddn{B} \sar \neg \xbox \ddn{B}, (\charx) [ \neg \ddn{B} \sar \ ]$}

\AxiomC{}
\RightLabel{$\botl$}
\UnaryInfC{$\ns$}
\RightLabel{$\iimpl$}
\BinaryInfC{$\neg \neg \xbox \ddn{B} \sar (\charx) [ \neg \ddn{B} \sar \bot ]$}
\RightLabel{$\iimpr$}
\UnaryInfC{$\neg \neg \xbox \ddn{B} \sar (\charx) [ \ \sar \neg \neg \ddn{B} ]$}

\AxiomC{$\sar (\charx) [ \neg \neg \ddn{B} \sar \ddn{B} ]$}
\RightLabel{$\cut$}
\BinaryInfC{$\neg \neg \xbox \ddn{B} \sar (\charx) [ \ \sar \ddn{B} ]$}
\RightLabel{$\xboxr$}
\UnaryInfC{$\neg \neg \xbox \ddn{B} \sar \xbox \ddn{B}$}
\DisplayProof
\end{center}
The remaining cases are analogous. 
\end{proof}

For $\Gamma = A_{1}, \ldots, A_{n}$, we define $\neg \Gamma := \neg A_{1}, \ldots, \neg A_{n}$. In the following proof, we make use of the \emph{swap-left operator} $\orth{}$, defined as: if $\ns$ is a CNS, then $\orth{\ns}$ is obtained by replacing each component $\Gamma \sar \Delta$ by $\Gamma, \neg \ddn{\Delta} \sar \emptyset$ in $\ns$, i.e., $(w,\Gamma \sar \Delta) \in \tr(\ns)$ \iffi $(w,\Gamma, \neg \ddn{\Delta} \sar \emptyset) \in \tr(\orth{\ns})$.

\begin{lemma}\label{lem:dn-classical-to-int}
If $\ns$ is provable in $\cglcalc$, then $\orth{\ns}$ is provable in $\calc$.
\end{lemma}

\begin{proof} By induction on the height of the proof of $\ns$ in $\cglcalc$. 

\textit{Base case.} We may transform an instance of $\id$ in $\cglcalc$ into a proof of the desired output in $\calc$. Recall that in the intuitionistic setting $\neg A = A \iimp \bot$, thus explaining the $\iimpl$ inference applied in the output proof. The top left leaf sequent in the output proof is easily shown to be provable as $p \sar \neg \neg p$ is valid in every IGL.
\begin{center}
\begin{tabular}{c c}
\AxiomC{}
\RightLabel{$\id$}
\UnaryInfC{$\ns\set{\Gamma, p \sar p, \Delta}$}
\DisplayProof

&

\AxiomC{$\orth{\ns}^{\da}\set{\Gamma, p, \neg \neg \neg p, \neg \ddn{\Delta} \sar \neg \neg p}$}
\AxiomC{}
\RightLabel{$\botl$}
\UnaryInfC{$\orth{\ns}\set{\Gamma, p, \bot, \neg \ddn{\Delta} \sar \emptyset}$}
\RightLabel{$\iimpl$}
\BinaryInfC{$\orth{\ns}\set{\Gamma, p, \neg \neg \neg p, \neg \ddn{\Delta} \sar \emptyset}$}
\DisplayProof
\end{tabular}
\end{center}

\textit{Inductive step.} We show the $\xboxr$ case below. The remaining cases are similar.
\begin{center}
\AxiomC{$\orth{\ns}\set{\Gamma, \neg \ddn{\Delta} \sar \emptyset, (\charx) [ \neg \ddn{A} \sar \emptyset ]}$}
\RightLabel{$\wkl, \wkr$}
\UnaryInfC{$\orth{\ns}\set{\Gamma, \neg \xbox \ddn{A}, \neg \ddn{\Delta} \sar \emptyset, (\charx) [ \neg \ddn{A} \sar \bot]}$}
\RightLabel{$\iimpr$}
\UnaryInfC{$\orth{\ns}\set{\Gamma, \neg \xbox \ddn{A}, \neg \ddn{\Delta} \sar \emptyset, (\charx) [ \emptyset \sar \neg \neg \ddn{A}]}$}

\AxiomC{$\nsii$}
\RightLabel{$\cut$}
\BinaryInfC{$\orth{\ns}\set{\Gamma, \neg \xbox \ddn{A}, \neg \ddn{\Delta} \sar \emptyset, (\charx) [\emptyset \sar \ddn{A}]}$}
\RightLabel{$\xboxr$}
\UnaryInfC{$\orth{\ns}\set{\Gamma, \neg \xbox \ddn{A}, \neg \ddn{\Delta} \sar \xbox \ddn{A}}$}

\AxiomC{}
\UnaryInfC{$\orth{\ns}\set{\Gamma, \bot, \neg \ddn{\Delta} \sar \emptyset}$}

\RightLabel{$\iimpl$}
\BinaryInfC{$\orth{\ns}\set{\Gamma, \neg \xbox \ddn{A}, \neg \ddn{\Delta} \sar \emptyset}$}
\DisplayProof
\end{center}
We let $ \nsii := \orth{\ns}\set{\Gamma, \neg \xbox \ddn{A}, \neg \ddn{\Delta} \sar \ (\charx) [ \neg \neg \ddn{A} \sar \ddn{A}]}$ denote the right premise of $\cut$, which is provable by \lem~\ref{lem:double-neg-implies-ddn-trans} and the hp-admissibility of $\nec$, $\ew$, $\wkl$, and $\wkr$. The top left leaf sequent is provable by IH.
\end{proof}

The following is straightforward to prove by induction on the length of $A$.

\begin{lemma}\label{lem:dn-implies-formula}
For any $A \in \lang$, $\ddn{A} \sar A$ is provable in $\cglcalc$.
\end{lemma}

In the following proof, we make use of the \emph{swap-right operator} $\north{}$, defined as follows: if $\ns$ is a nested sequent, then $\north{\ns}$ is the CNS obtained by replacing each component $\Gamma \sar \Delta$ by $\sar \neg \Gamma, \Delta$ in $\ns$, that is, $(w,\Gamma \sar \Delta) \in \tr(\ns)$ \iffi $(w, \emptyset \sar \neg \Gamma, \Delta) \in \tr(\north{\ns})$.

\begin{lemma}\label{lem:dn-int-to-classical}
If $\ns$ is provable in $\calc$, then $\north{\ns}$ is provable in $\cglcalc$.
\end{lemma}

\begin{proof} By induction on the height of the proof of $\ns$. The base cases for the rules $\id$ and $\botl$ are trivial. We show the $\xdial$ case of the inductive step and note that the remaining cases are straightforward.
\begin{center}
\AxiomC{$\ns \sbl \emptyset \sar \neg \Gamma, \Delta, (\charx) [\emptyset \sar \neg A ] \sbr$}
\RightLabel{$\xboxr$}
\UnaryInfC{$\ns \sbl \emptyset \sar \neg \Gamma, \xbox \neg A, \Delta \sbr$}

\AxiomC{$\ns \sbl \xbox \neg A \sar \neg \Gamma, \neg \xdia A, \Delta \sbr$}

\RightLabel{$\cut$}
\BinaryInfC{$\ns \sbl \emptyset \sar \neg \Gamma, \neg \xdia A, \Delta \sbr$}
\DisplayProof
\end{center}
The right premise of $\cut$ is provable since $\ns \sbl \xbox \neg A \sar \neg \xdia A \sbr$ is valid in every CGL. The top left leaf sequent is provable by IH. The result follows since $\cut$ is eliminable in $\cglcalc$ by~\cite[Theorem 3.4]{TiuIanGor12}.
\end{proof}

\begin{theorem}\label{thm:classical-intuitionistic-reduction}
$\cglcalc \proves (\emptyset \sar A)$ \iffi $\calc \proves (\emptyset \sar \ddn{A})$.
\end{theorem}

\begin{proof} For the forward direction, assume $\emptyset \sar A$ has a proof in $\cglcalc$. By \lem~\ref{lem:dn-classical-to-int}, $\neg \ddn{A} \sar \emptyset$ has a proof in $\calc$, meaning, $\neg \ddn{A} \sar \bot$ has a proof by the hp-admissibility of $\wkr$. Therefore, $\emptyset \sar \neg \neg \ddn{A}$ has a proof by one application of $\iimpr$. By \lem ~\ref{lem:double-neg-implies-ddn-trans} and the admissibility of $\cut$, we know that $\emptyset \sar \ddn{A}$ has a proof in $\calc$. 

For the backward direction, we assume that $\emptyset \sar \ddn{A}$ has a proof in $\calc$. By \lem~\ref{lem:dn-int-to-classical}, $\emptyset \sar \ddn{A}$ is provable in $\cglcalc$. By \lem~\ref{lem:dn-implies-formula} and $\cut$ admissibility, it follows that $\sar A$ is provable in $\cglcalc$.
\end{proof}

The following is a consequence of the above theorem and the undecidability of the general validity problem for classical context-free grammar logics~\cite{BalGioMar98}.

\begin{corollary}
The general validity problem is undecidable for our class of IGLs.
\end{corollary}

\section{Concluding Remarks}\label{sec:conclusion}


In this paper, we initiated the structural proof theory of intuitionistic grammar logics by supplying proof systems in the formalism of nested sequents. We investigated fundamental invertibility and admissibility properties, and introduced the novel shift $\shift$ rule, which enabled a uniform proof of cut-admissibility across all IGLs. Completeness was established syntactically by means of cut-admissibility, while soundness was proven model-theoretically. Finally, we adapted the negative translation of G\"{o}del, Gentzen, and Kolmogorov to the multi-modal setting, showing that a formula $A$ is a theorem of a CGL if and only if its translation $\ddn{A}$ is a theorem of the corresponding IGL. This equivalence is witnessed by proof transformations between the multi-conclusioned nested sequent proofs of Tiu et al.~\cite{TiuIanGor12} and our single-conclusioned nested sequent proofs, and reduces the general validity problem for CGLs to that of IGLs. Since the former is known to be undecidable, this yields the undecidability of the general validity problem for IGLs as well.

There are several directions for future research. For instance, our nested sequent systems are not fully invertible, e.g., the $\iimpl$ rule is not invertible. This raises the natural goal of developing fully invertible systems for IGLs, which may be better suited for counter-model extraction and, consequently, for investigating decidability. It seems plausible that labeled generalizations of our nested systems, in the style of Marin et al.~\cite{MarMorStr21}, would yield fully invertible calculi. Such systems may also extend naturally to Scott--Lemmon axioms, since labeled calculi are well suited to capture their corresponding frame conditions. Moreover, since the general validity problem for IGLs is undecidable, a more fine-grained analysis of which subclasses of IGLs are decidable is warranted, together with a study of the computational complexity of those decidable fragments.


\bibliography{bibliography}

\end{document}